\documentclass[english,10pt]{llncs}

\usepackage{etex}
\reserveinserts{28}
\usepackage[T1]{fontenc}
\usepackage[utf8x]{inputenc}
\usepackage[english]{babel}
\usepackage{graphics, latexsym,amssymb,amsmath}

\usepackage{siunitx}
\usepackage{booktabs}
\usepackage{xfrac}
\newcommand{\ra}[1]{\renewcommand{\arraystretch}{#1}}

\usepackage{graphicx}
\usepackage{pgf}
\usepackage{tikz}
\usepackage{tikz-qtree}

\definecolor{mybrown}{RGB}{33,34,28}
\definecolor{myyellow}{RGB}{242,226,149}
\definecolor{mygreen}{RGB}{176,232,145}
\definecolor{myblue}{RGB}{61,139,189}
\definecolor{myorange}{RGB}{245,156,74}
\definecolor{mypurple}{RGB}{230,111,148}
\definecolor{myred}{RGB}{215,80,50}

\tikzstyle{convex1trans}=[color=black,opacity=0.55]
\tikzstyle{convex1}=[color=black!55]
\tikzstyle{convex2}=[color=black!20]
\newcommand{\PSVR}{$\mathcal{PSVR}$\xspace}
\newcommand{\ps}{pub/sub\xspace}

\tikzstyle{treenode}=[draw,circle,minimum width=12pt,inner sep=0pt,outer sep=0pt]

\tikzstyle{root}=[treenode]
\tikzstyle{l1}=[treenode,rectangle,minimum height=12pt]
\tikzstyle{l2}=[treenode]

\newcommand{\deltaS}{\delta_S\xspace}

\usetikzlibrary{matrix,arrows,backgrounds,calc,trees,shapes,patterns}

\makeatletter
\tikzset{circle split part fill/.style  args={#1,#2}{%
 alias=tmp@name, 
  postaction={%
    insert path={
     \pgfextra{%
     \pgfpointdiff{\pgfpointanchor{\pgf@node@name}{center}}%
                  {\pgfpointanchor{\pgf@node@name}{east}}%
     \pgfmathsetmacro\insiderad{\pgf@x}
      \fill[#1] (\pgf@node@name.base) ([xshift=-\pgflinewidth]\pgf@node@name.east) arc
                          (0:180:\insiderad-\pgflinewidth)--cycle;
      \fill[#2] (\pgf@node@name.base) ([xshift=\pgflinewidth]\pgf@node@name.west)  arc
                           (180:360:\insiderad-\pgflinewidth)--cycle;            
         }}}}}  
 \makeatother  

\pgfdeclarelayer{background}
\pgfsetlayers{background,main}

\makeatletter
\tikzset{
        hatch distance/.store in=\hatchdistance,
        hatch distance=5pt,
        hatch thickness/.store in=\hatchthickness,
        hatch thickness=5pt
        }
\pgfdeclarepatternformonly[\hatchdistance,\hatchthickness]{north east hatch}
    {\pgfqpoint{-1pt}{-1pt}}
    {\pgfqpoint{\hatchdistance}{\hatchdistance}}
    {\pgfpoint{\hatchdistance-1pt}{\hatchdistance-1pt}}%
    {
        \pgfsetcolor{\tikz@pattern@color}
        \pgfsetlinewidth{\hatchthickness}
        \pgfpathmoveto{\pgfqpoint{0pt}{0pt}}
        \pgfpathlineto{\pgfqpoint{\hatchdistance}{\hatchdistance}}
        \pgfusepath{stroke}
    }

\pgfdeclarepatternformonly[\hatchdistance,\hatchthickness]{north west hatch}
    {\pgfqpoint{-1pt}{-1pt}}
    {\pgfqpoint{\hatchdistance}{\hatchdistance}}
    {\pgfpoint{\hatchdistance-1pt}{\hatchdistance-1pt}}%
    {
        \pgfsetcolor{\tikz@pattern@color}
        \pgfsetlinewidth{\hatchthickness}
        \pgfpathmoveto{\pgfqpoint{0pt}{0pt}}
        \pgfpathlineto{\pgfqpoint{\hatchdistance}{\hatchdistance}}
        \pgfusepath{stroke}
    }
\makeatother

\usepackage{mdframed}
\usepackage{xcolor}

\usepackage{xcolor,colortbl}

\mdfsetup{%
skipabove=1pt, leftmargin=1pt, skipbelow=1pt,rightmargin=1pt,innerleftmargin=1pt,innerrightmargin = 1pt,innertopmargin = 1pt,innerbottommargin = 1pt, linewidth=0}

\newcommand*\circled[1]{\tikz[baseline=(char.base)]{
            \node[shape=circle, minimum size = 0.39cm, draw,inner sep=0.5pt] (char) {#1};}}

\usepackage{color}
\usepackage{algorithm}
\usepackage[noend]{algpseudocode}
\graphicspath{{.}{./Figures}}
\usepackage{multirow}
\usepackage{cite}
\usepackage{multicol}
\usepackage{xspace}
\usepackage[colorinlistoftodos,shadow]{todonotes}
\usepackage{wrapfig}
\usepackage{caption}
\usepackage{subcaption}

\definecolor{pinegreen}{cmyk}{0.92,0,0.59,0.25}
\definecolor{royalblue}{cmyk}{1,0.50,0,0}
\definecolor{lavander}{cmyk}{0,0.48,0,0}
\definecolor{violet}{cmyk}{0.79,0.88,0,0}
\tikzstyle{follow}=[circle, draw, very thin,fill=white]
\tikzstyle{lead}=[circle, draw, thick,fill=white]
\tikzstyle{lead2}=[circle, draw, thick,fill=white, densely dashed]
\tikzstyle{lead3}=[circle, draw, thick,fill=white, densely dotted]
\tikzstyle{lead4}=[circle, draw, thick,fill=white, densely dashdotted]
\tikzstyle{hidden}=[circle, draw, opacity=0, densely dashed]

\newcommand{\publish}[1][]{\emph{{publish}(#1)}\xspace}
\newcommand{\subscribe}[1][]{\emph{{subscribe}(#1)}\xspace}
\newcommand{\sendsub}[1][]{\emph{broadcast(#1)}\xspace}

\newcommand{\getPosClosestTo}[1][]{\emph{\textbf{getPosClosestTo}(#1)}\xspace}

\newcommand{\isBetween}[1][]{\emph{\textbf{isBetween}(#1)}\xspace}
\newcommand{\sendRing}[1][]{\emph{\textbf{sendOnRing}(#1)}\xspace}

\newcommand{\borderPos}[1][]{ep\xspace}
\newcommand{\nothing}[1]{}

\newcommand{\expireTimer}[0]{t_s\xspace}

\newcommand{\OMNET}{\mbox{OMNeT++}\xspace}
\newcommand{\MIXIM}{\mbox{MiXiM}\xspace}

\tikzstyle{legend_isps}=[rectangle, rounded corners, thin, 
fill=gray!20, text=blue, draw]

\tikzstyle{legend_overlay}=[rectangle, rounded corners, thin,
top color= white,bottom color=green!25,
minimum width=2.5cm, minimum height=0.8cm,
pinegreen]
\tikzstyle{legend_phytop}=[rectangle, rounded corners, thin,
top color= white,bottom color=cyan!25,
minimum width=2.5cm, minimum height=0.8cm,
royalblue]
\tikzstyle{legend_general}=[rectangle, rounded corners, thin,
top color= white,bottom color=lavander!25,
minimum width=2.5cm, minimum height=0.8cm,
violet]

\newcommand{\fwdTable}[1][]{RS(v)\xspace}
\newcommand{\fwdTableWOUTv}[1][]{RS\xspace}

\newcommand{\parent}[1][]{$par$\xspace}
\newcommand{\child}[1][]{Chd\xspace}
\newcommand{\neighTableSize}[1][]{C_{N}\xspace}
\newcommand{\numOChannels}[1][]{C_{c}\xspace}
\newcommand{\neighTable}[1][]{$N$\xspace}
\newcommand{\positionTable}[1][]{R\xspace}
\newcommand{\ownPosVec}[1][]{$P$\xspace}

\DeclareMathAlphabet\mathbfcal{OMS}{cmsy}{b}{n}

\usepackage{xpatch}
\makeatletter
\xpatchcmd{\algorithmic}{\ALG@tlm\z@}{\ALG@tlm\z@\leftmargin 0pt}{}{}
\makeatother

\begin{document}
\title{$\mathcal{PSVR}$ - Self-stabilizing Publish/Subscribe Communication for Ad-hoc Networks }
\date{2015}


\author{%
	  G. Siegemund \and  V. Turau
}
\institute{%
	Institute of Telematics, 
	Hamburg University of Technology, Hamburg, Germany.
	\email{\{gerry.siegemund,turau\}@tuhh.de} 
}

\maketitle

\begin{abstract}
  This paper presents the novel routing algorithm $\mathcal{PSVR}$ for
  pub/sub systems in ad-hoc networks. Its focus is on scenarios where
  communications links are unstable and nodes frequently change
  subscriptions. $\mathcal{PSVR}$ presents a compromise of size and
  maintenance effort for routing tables due to sub- and
  unsubscriptions and the length of routing paths. Designed in a
  self-stabilizing manner it scales well with network size. The
  evaluation reveals that $\mathcal{PSVR}$ only needs slightly more
  messages than a close to optimal routing structure for publication
  delivery, and creates shorter routing paths than an existing
  self-stabilizing algorithm. A real world deployment shows the
  usability of the approach.
\end{abstract}


\section{Introduction}
Industrial wireless sensor networks are an emerging field for process
monitoring and control that require dynamic forms of the many-to-many
communication paradigm for data dissemination. This communication
style is best supported by publish/subscribe (pub/sub) systems instead of
request-reply messaging. In channel-based pub/sub systems, publishers
assign each message to one of several channels which are known by all
nodes. Subscribers express interest in one or more channels (a.k.a.\
subscribing to the channel) and only receive messages assigned to
these. The pub/sub paradigm guarantees disseminating all
messages to nodes with a subscription for that channel. The
advantage is the loose coupling, i.e., publishers are unaware of the
subscribers that receive their messages. Nodes can at any time
give up subscriptions and create new ones.

The efficiency of message dissemination in pub/sub systems depends on
the used routing strategy. The goal is to deliver each publication
with a minimum number of messages to all subscribers. The minimum
number of messages is used when the publication is routed along the
Steiner tree for the publishing node and all nodes subscribing to the
message's channel. Since Steiner trees are computationally too
expensive many systems use a fixed spanning tree for routing. A
publisher recursively forwards a message into those subtrees that
contain a subscriber for the message's channel. This requires each
node to provide the necessary information and does in general not
result in the shortest routing path. Other systems organize their
nodes into a virtual ring. A published message is then simply
forwarded once around this ring and thereby delivered to all
subscribers. This does not require any routing tables and there is no
need to distribute un-/subscriptions into the network. Unfortunately
this requires at least as many messages as nodes in the virtual ring.

In this paper we consider scenarios where nodes frequently change
their subscriptions, hence, an efficient update of the routing
structure is required. Also delivery of publications must be
guaranteed while subscriptions are changing. To meet this goal we
propose the routing algorithm $\mathcal{PSVR}$, which is a significant
extension of the algorithm in \cite{Siegemund_VR:2015}.
$\mathcal{PSVR}$ presents a compromise between the length of routing
paths and the effort to maintain the routing tables. One of the core
ideas is to augment routing on the virtual ring by shortcuts. We show
that for a specific class of graphs on average the increase of the
length of routing paths is bearable and updating a node's subscription
list is simple. To increase system robustness and to tolerate the
failure and recovery of links and nodes the proposed algorithms are
self-stabilizing. The effectiveness of the proposed algorithm is shown
through simulations using a realistic channel model and by a
comparison with a self-stabilizing tree-based approach.

\section{Related Work}
\label{sec:soa}
The general state of the art for pub/sub systems for WSN is summarized
in a recent survey \cite{Sheltami:2015}. Detti et al.\ classify
pub/sub systems into pull and push systems \cite{Detti:2015}. In the
first model nodes interested in a channel periodically flood the network
with interest messages upon which nodes respond with cached
publications for this channel via reverse paths. The number of sent
messages is dominated by the frequency of of issued interest messages
-- which reflects latency -- and not by the number of subscriptions.
Such an approach is of advantage in mobile environments where routing
structures are quickly outdated. In the push model the number of
messages sent mainly depends on the rate of publications and the
number of subscribers, given a routing structure. For static
environments this approach is of advantage. Baldoni et al.\
distinguish between message and subscription forwarding
\cite{baldoni:2005}. In the former case all publications are forwarded
via a fixed spanning tree. Thus, the number of forwarded messages does
not scale with the number of network nodes. Subscription forwarding
permits to establish a routing structure that allows to forward
publications to subscribers only. This way the number of forwarded
messages is independent of the total number of nodes but depends on
the number of subscribers and their positions.


Message forwarding is mainly of interest if the number of subscribers
is large compared to the number of nodes and if the number of
publications is low. Mires~\cite{souto2006mires} is a pub/sub
middleware for WSNs where the sink is the sole subscriber. Nodes
advertise the data they can provide and the sink thereafter informs
nodes about its interest. Routing is performed along a fixed tree.
Fault tolerance is not considered. Proposed standards such as MQTT-S
\cite{Hunkeler:2008} and DDS 
\cite{pardo2003omg} mainly address QoS and are not
tailored towards resource constrained networks. Pub/sub systems such
as Scribe use overlay networks based on distributed hash tables
\cite{Castro:2006}. Overlay networks are logical networks on top of
real network where links correspond to paths in the underlying
network, which are usually IP-based networks. Thus, this approach is
unsuitable for WSNs. One of the first pub/sub system dedicated to
wireless ad-hoc networks is described in \cite{Huang:2003}. A greedy
algorithm builds a tree for each node which is used to route
publications to subscribers. Fault tolerance is not addressed.

The first proposal for subscription forwarding is directed diffusion
\cite{Intanagonwiwat:2003}. Each subscription sets up gradients in
the network, these are used to deliver publications. Even so nodes
cache information from previous subscriptions the message overhead is
high. Negative reinforcement is used to eliminate loops. Directed
diffusion provides some degree of fault tolerance by maintaining
alternative paths.


We are only aware of two self-stabilizing pub/sub systems
\cite{Jaeger:2008,Shen:2007}. Jaeger's system uses a broker overlay
network to route publications to subscribers connected to brokers,
these only forward the data \cite{Jaeger:2008}. Subscriptions and
advertisements are used to generate routing tables. The leasing
technique is used to fix possible faults in these routing tables. The
renewal of leases is triggered by periodically dispensed subscription
messages, an expired lease leads to the removal of the entry. Shen
uses a spanning tree to route publications \cite{Shen:2007}. Nodes
maintain routing tables to forward publications. To provide fault
tolerance routing tables are exchanged periodically. This mechanism
cannot tolerate all types of faults, e.g., the concurrent loss of
routing entries in several nodes. A self-repairing content-based
routing algorithm is described in \cite{Mottola:2008}.

A disadvantage of all tree-based routing approaches is that only the
$n-1$ communication links of the tree are used \cite{Chen:2013}. For
dense networks this excludes the majority of links and leads to long
routing paths. To circumvent this disadvantage a self-stabilizing
pub/sub system based on a virtual ring is introduced in
\cite{Siegemund_VR:2015}. A virtual ring is a directed closed path
over all nodes. It allows for a very simple dissemination of
publications without requiring knowledge of the topology, but
forwarding paths can be much longer than the shortest paths. The
remedy used in \cite{Siegemund_VR:2015} is to use edges that are not
part of the ring as short-cuts. The result is a compromise between the
complexity of the routing tables and the lengths of the forwarding
paths. A positive aspect is that it is easy to adapt the structure to
new subscribers, but the approach of \cite{Siegemund_VR:2015} has
several shortcomings. Firstly, nodes may receive a subscription
several times. Secondly, subscription messages are forwarded to all
subscribers. Also publications are not discarded by the last
subscriber on the ring but sent further along the ring.
Unsubscriptions are only marginally addressed
in~\cite{Siegemund_VR:2015}. If a stale routing table entry is not
refreshed within the leasing period, it is removed and all further
messages are routed along the virtual ring instead. This leads to a
temporary loss of routing information, hence, to longer routing paths.

\section{Foundation}
\label{sec:foundation}

Let $G=(V,E)$ be an undirected graph with $n$ nodes. A virtual ring is
a closed path over all nodes formally defined as follows.

\begin{definition} 
  A sequence \mbox{$R=\langle v_0, \dots, v_{l-1} \rangle$} of nodes
  $v_i \in V$ is called a \textit{virtual ring} if each $v\in V$
  appears at least once in $R$ and if each $v_i$ is a neighbor of
  $v_{i+1}$ (indices are taken modulo $l$); $l$ is called the
  \textit{length} of $R$. For $v \in V$ each $i$ with $v=v_i$ is
  called a \textit{position} of $v$. The list of positions of $v$ is
  denoted by $Pos(v)$.
\label{def:vr}
\end{definition}

Each connected graph possesses a virtual ring. Note that $l = \sum_{v
  \in V} |Pos(v)|$. For a virtual ring of short length the sets
$Pos(v)$ must be small. Only Hamiltonian graphs have virtual rings
with $|Pos(v)| = 1$ for each $v \in V$ (i.e., $l=n$). A depth-first
traversal of a tree $T$, where every node visit is recorded with an
incremented value, determines a virtual ring $R$ and all node
positions. A node $v$ has as many positions on $R$ as $v$ has
neighbors in $T$, i.e., $l=2(n-1)$. Figure~\ref{fig:stack_small} shows
a spanning tree (bold edges) for a topology with six nodes (left).

\vspace{-5mm}
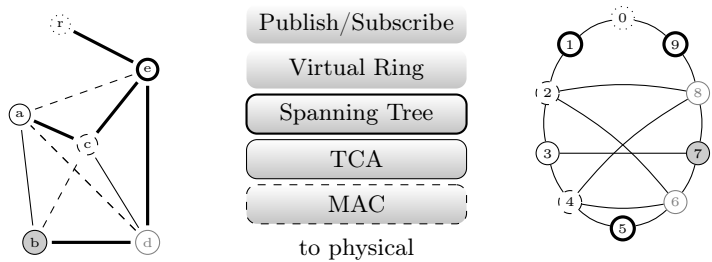
\begin{figure}[h]
\centering
\small
~\\
~
~\\
~\\
\begin{tikzpicture}
\tikzset{
    mynode/.style={rectangle,align=center, rounded corners,draw=black, top color=white, bottom color=gray!45!black!30,	inner sep=0.2em, minimum size=1.6em, minimum width=9em, text centered},
    emptynode/.style={rectangle,align=center, rounded corners,draw=white, top color=white, bottom color=white,
    					thin, inner sep=0.2em, minimum size=1.6em, minimum width=9em, text centered},
    myarrow/.style={<->, >=latex', shorten >=1pt, thick},
    mylabel/.style={text width=7em, text centered},
    dotA/.style={circle, draw=black, font=\tiny, inner sep=0.5mm}, %
    edgeA/.style={shorten >= 0.6mm, shorten <= 0.6mm, dashed}, %
    edgeChosen/.style={shorten >= 0.6mm, shorten <= 0.6mm, very thick}, %
    edgeTree/.style={shorten >= 0.6mm, shorten <= 0.6mm, thin}, %
    dotVR/.style={circle, draw=black, font=\tiny, inner sep=0.5mm}, %
    edgeVR/.style={shorten >= 0.2mm, shorten <= 0.2mm, bend left},%
    connecto/.style={shorten >= 0.5mm, shorten <= 0.5mm}%

}  

		
\node[mynode, draw=none] at (0,10) (l1) {Publish/Subscribe};  
\node[mynode, draw=none] at (0,9.4) (l2) {Virtual Ring};  
\node[mynode, thick] at (0,8.8) (l3) {Spanning Tree};  
\node[mynode, thin] at (0,8.2) (l4) {TCA};  
\node[mynode, dashed] at (0,7.6) (l5) {MAC};  
\node[emptynode] at (0,7) (l6) {to physical};  
 
\node[dotA, dotted] (d0) at (-3.9364,9.9904) {r};
\node[dotA] (d1) at (-4.4774,8.7932) {a};
\node[dotA, very thick] (d5) at (-2.7774,9.3932) {e};
\node[dotA, dashed] (d3) at (-3.5774,8.3932) {c};
\node[dotA, fill=gray!40!] (d2) at (-4.2774,7.0932) {b};
\node[dotA, gray] (d4) at (-2.7774,7.0932) {d};


\draw[edgeChosen] (d0) -- (d5);
\draw[edgeChosen] (d5) -- (d3);
\draw[edgeChosen] (d5) -- (d4);
\draw[edgeA] (d5) -- (d1);
\draw[edgeTree] (d4) -- (d3);
\draw[edgeChosen] (d4) -- (d2);
\draw[edgeA] (d4) -- (d1);
\draw[edgeA] (d2) -- (d3);
\draw[edgeA] (d4) -- (d1);
\draw[edgeChosen] (d1) -- (d3);
\draw[edgeTree] (d1) -- (d2);

\node[dotVR, dotted] (v0) at (3.5399,10.0791) {0};
\node[dotVR, very thick] (v1) at (2.8399,9.7291) {1};
\node[dotVR, dashed] (v2) at (2.5399,9.0791) {2};
\node[dotVR] (v3) at (2.5399,8.2791) {3};
\node[dotVR, dashed] (v4) at (2.8399,7.6291) {4};
\node[dotVR, very thick] (v5) at (3.5399,7.2791) {5};
\node[dotVR, gray] (v6) at (4.2399,7.6291) {6};
\node[dotVR, fill=gray!40!] (v7) at (4.5399,8.2791) {7};
\node[dotVR, gray] (v8) at (4.5399,9.0791) {8};
\node[dotVR, very thick] (v9) at (4.2399,9.7291) {9};

\path[] (v0)  edge [bend right = 15] (v1);
\path[] (v1)  edge [bend right = 10] (v2);
\path[] (v2)  edge [bend right = 10] (v3);
\path[] (v3)  edge [bend right = 10] (v4);
\path[] (v4)  edge [bend right = 15] (v5);
\path[] (v5)  edge [bend right = 15] (v6);
\path[] (v6)  edge [bend right = 10] (v7);
\path[] (v7)  edge [bend right = 10] (v8);
\path[] (v8)  edge [bend right = 10] (v9);
\path[] (v9)  edge [bend right = 15] (v0);

\path[] (v3)  edge [] (v7);
\path[] (v2)  edge [bend left = 10] (v8) edge [bend left = 10] (v6);
\path[] (v4)  edge [bend left = 10] (v8) edge [bend right = 10] (v6);

\end{tikzpicture}

\caption{Topology; layered system architecture; corresponding virtual ring graph.}
\label{fig:stack_small}
\end{figure}
\vspace{-3mm}
As in \cite{Siegemund_VR:2015} we use a topology control algorithm
(TCA) to mark a communication graph using only high quality,
bi-directional, stable links. The chosen TCA is dynamic, deteriorating
links are removed, while new promising links are added. The number of
maintained neighbors is limited to $\neighTableSize$ to accommodate
restricted memory resources. The TCA acts as a message filter for
broadcasted messages. Messages received from nodes that are not in the
current neighbor set are not dispatched to upper layers. Thus, links
not chosen by the TCA are transparent to upper layers. In
Fig.\ref{fig:stack_small} (left) edges selected by the TCA are
depicted as solid lines ($\neighTableSize=3$) while dashed edges were
excluded. To benefit from links selected by the TCA that are not part
of the virtual ring, \textit{shortcuts} are introduced.



\begin{definition}
An edge $(v_i , v_j)$ with $j \neq i+1$ is called a \textit{shortcut} in a ring $R$.
\label{def:shortcut}
\end{definition}


In the following a virtual ring based on depth-first traversal is
interpreted as a graph $G_R$ where the nodes correspond to the
positions of the original nodes. If $G$ has $n$ nodes, the virtual
ring $G_R$ graph has $2(n-1)$ nodes. The edges of $G_R$ correspond to
the edges of $R$ and the shortcuts of~$R$ in~$G$.
Figure~\ref{fig:stack_small} (right) shows the virtual ring graph
emerging from the given topology, the selection conducted by the TCA,
and the spanning tree on the left, i.e., \mbox{$R=\langle
  r,e,c,a,c,e,d,b,d,e\rangle$}. The edge between $c$ and $d$ in the
topology, results in the shortcuts between positions 2, 4 and 6, 8 in
$G_R$. A node's representation in the topology corresponds to the
appearance of its position on $G_R$.


With the virtual ring and the shortcuts in place, the pub/sub routing
algorithm can be explained. Nodes can take the role of publishers,
subscribers, both or none. Independent of their role, nodes forward
messages via links of the virtual ring graph. Two message types are
used: \textit{subscriptions} to build and update routing tables and
\textit{publications} to carry the data. Routing on each channel
is independent. The creation of channels is not explicitly stated
in~\cite{Siegemund_VR:2015}. In the following we assume that channels
are defined prior to system start-up, and their existence is known to
all nodes. Hereafter, since channels are independent of each other, if
not stated otherwise only a single channel is considered.

A trivial way to route publications on the virtual ring is to
consecutively hand them to each successor and to deliver them if a
corresponding subscription exists. When a publication returns to its
originator it is discarded. With the virtual ring in place, routing
tables are trivial. Even though this procedure is simple and
memory-conserving, nodes with multiple positions receive publications
repeatedly. Each message is forwarded $l$ times, i.e., independent of
the number of subscribers. This decreases robustness due to the
increased message loss probability and increases latency. Furthermore,
each node receives all publications in the network regardless of being
a subscriber or not. This trivial routing scheme is significantly
improved in~\cite{Siegemund_VR:2015} by using shortcuts. These lead on
average to shorter routing paths. Using the leasing technique, it is
shown that the system is self-stabilizing and therefore inherent
fault-tolerant. Nevertheless the approach is flawed, shortcomings in
every section of the pub/sub system have been identified and solutions
to those issues are presented next.

\section{$\mathcal{P}$ublish/$\mathcal{S}$ubscribe on $\mathcal{V}$irtual $\mathcal{R}$ings}
\label{sec:PSVRoverview}

The architecture of $\mathcal{PSVR}$ is shown in
Fig.~\ref{fig:stack_small}. For details about the virtual ring, the
spanning tree, and the TCA we refer to~\cite{Siegemund_VR:2015}. The
spanning tree layer is slightly augmented to enhance the dissemination
of subscriptions.




\subsubsection*{Routing tables in $\mathcal{PSVR}$.}
Each node $v$ maintains a routing structure $\fwdTable$ in form of a
$n_c \times n_p$ matrix, $n_c$ denotes the number of channels and
\mbox{$n_p = |Pos(v)|$}. $\fwdTableWOUTv$ stores tuples in the form
$\langle ns, \expireTimer, nstmp \rangle$. When a message for the $c_i
{^{th}}$ channel is received at the $p_j {^{th}}$ position, then
$\fwdTable[] [c_i,p_j].ns$ is the position of the subscriber for
channel $c_i$ which is counter clock wise (ccw) closest to the $p_j
{^{th}}$ position (called forwarding position). The components
$\expireTimer, nstmp$ are used for unsubscriptions (see
Sec.~\ref{sec:unSub}). Before described the routing of publications
the novel subscription dissemination on the pub/sub and on the tree
layer is presented.


\subsection{Subscriptions}
\label{sec:subOnTree}

Subscription messages are used to maintain the routing structures
$\fwdTableWOUTv$ at all nodes. Lost subscription messages do not lead
to a permanent omission of publications, because the leasing technique
guarantees the renewal of a subscription within time $\delta_S$.

\subsubsection*{Subscription Distribution Range.}
In~$\fwdTable$ the next ccw subscriber for each position of $v$ is
stored. 
A newly subscribing node $w$ requires that nodes update their
routing structure. In particular a node $u$ needs to update
$\fwdTableWOUTv{}(u)$ if and only if there exists a position $p_w \in
Pos(w)$ and a position $p_u \in Pos(u)$ such that $p_w$ is ccw in
between $p_u$ and $p_u^f$, where $p_u^f$ is the according forwarding
position in $\fwdTableWOUTv{}(u)$ for $p_u$. That is, only the
positions between a new subscriber~$w$ and the \textit{clock wise}
closest subscriber~$u$, i.e., all nodes in the interval $[u,w)$, need
to receive subscriptions from~$w$.
Figure~\ref{fig:subscribtionDist}~shows the stored next subscriber.
Positions in the interval $[7,9)$ record position~9 as next
subscriber. All other positions, i.e., the positions in the interval
$[9,7)$, store position~7.

\vspace*{-2mm}
\begin{figure}[h]
    \centering
    \usetikzlibrary{backgrounds}
\begin{tikzpicture}
\tikzset{
    mynode/.style={rectangle,align=center, rounded corners,draw=black, top color=white, bottom color=gray!45!black!30,	inner sep=0.2em, minimum size=1.6em, minimum width=9em, text centered},
    myarrow/.style={<->, >=latex', shorten >=1pt, thick},
    msgEdge/.style={-, >=latex', shorten >=3pt},
    mylabel/.style={text width=7em, text centered},
    dotA/.style={circle, draw=black, font=\tiny, inner sep=0.5mm}, %
    pub/.style={circle, draw=black, fill=black!60!gray, text = white, font=\tiny, inner sep=0.5mm}, %
    sub/.style={minimum size=3.8mm,circle, draw=black, fill=black!10!gray!30, font=\tiny, inner sep=0.5mm}, %
	subT/.style={densely dotted, minimum size=3.8mm,circle, draw=black, fill=black!10!gray!30, font=\tiny, inner sep=0.5mm}, %
	subO/.style={thick, minimum size=3.8mm,circle, draw=black, fill=black!10!gray!30, font=\tiny, inner sep=0.5mm}, %
    edgeA/.style={shorten >= 0.6mm, shorten <= 0.6mm, dashed}, %
    edgeChosen/.style={shorten >= 0.6mm, shorten <= 0.6mm, very thick}, %
    edgeTree/.style={shorten >= 0.6mm, shorten <= 0.6mm, thin}, %
    dotVR/.style={fill = white, inner sep=-2em, minimum size=3.8mm, circle, draw=black, font=\tiny, inner sep=0.5mm}, %
    edgeVR/.style={shorten >= 0.2mm, shorten <= 0.2mm, bend right=40, densely dotted},%
    connecto/.style={shorten >= 0.5mm, shorten <= 0.5mm}%

}

\node[] (v2) at (-0.3,7.02) {};
\node[dotVR, label={[font=\tiny]270:$ $}] (v3) at (0.1,6.86) {3};
\node[dotVR, label={[font=\tiny]268:$ $}] (v4) at (0.85,6.65) {4};
\node[dotVR, label={[font=\tiny]180:$ $}] (v5) at (1.79,6.49) {5};
\node[dotVR, label={[font=\tiny]270:$ $}] (v6) at (2.78,6.42) {6};
\node[sub, label={[font=\tiny]270:$ $}] (v7) at (3.71,6.42) {7};
\node[dotVR] (v8) at (4.7,6.5) {8}; 
\node[sub] (v9) at (5.58,6.65) {9};
\node[dotVR] (v10) at (6.3,6.78) {10};
\node[] (v11) at (6.8409,6.9312) {};

\path[->] (v2)  edge [bend right = 5]  (v3);
\path[->] (v3)  edge [bend right = 5] (v4);
\path[->] (v4)  edge [bend right = 5] (v5);
\path[->] (v5)  edge [bend right = 5] (v6);
\path[->] (v6)  edge [bend right = 5] (v7);
\path[->] (v7)  edge [bend right = 5] (v8);
\path[->] (v8)  edge [bend right = 5] (v9);
\path[->] (v9)  edge [bend right = 5] (v10);
\path[-] (v10)  edge [bend right = 5] (v11);

\path[->] (v4)  edge [bend left = 10] node[below] {}  (v9);

\node[label={[font=\tiny]270:}] at (-0.91,6.5) {};
\node[label={[font=\tiny]270:RS.ns}] at (-0.56,6.32) {};
\node[label={[font=\tiny]270:$7 $}] at (0.1,6.32) {};
\node[label={[font=\tiny]270:$7 $}] at (0.85,6.32) {};
\node[label={[font=\tiny]270:$7 $}] at (1.79,6.32) {};
\node[label={[font=\tiny]270:$7 $}] at (2.78,6.32) {};
\node[label={[font=\tiny]270:$9 $}] at (3.71,6.32) {};
\node[label={[font=\tiny]270:$9 $}] at (4.7,6.32) {};
\node[label={[font=\tiny]270:$7 $}] at (5.58,6.32) {};
\node[label={[font=\tiny]270:$7 $}] at (6.3,6.32) {};


\begin{scope}[on background layer]
\draw[gray!50, double=gray!15,double distance=5mm,smooth,line cap=round,tension=0.4] plot coordinates {(v7) (v8) };
\draw[gray!50, double distance=5mm,smooth,line cap=round,tension=0.4] plot coordinates {(v9)(v10) (v11)};
\draw[gray!50, double distance=5mm,smooth,line cap=round,tension=0.4] plot coordinates {(v2) (v3)(v4)(v5)(v6)};
\end{scope}

\node[circle, draw, white, fill=white,minimum size=7mm] at (-0.59,7.1) {};
\node[circle, draw, white, fill=white,minimum size=7mm] at (7.14,7) {};
\end{tikzpicture}
    \caption{Virtual ring with two subscribers (gray).}
    \label{fig:subscribtionDist}
\end{figure}
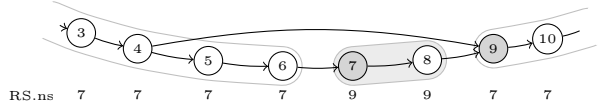
\vspace*{-5mm}

\subsubsection*{Distributed Subscription Routing over the Tree.}
Subscription messages can be spread faster and with viewer messages if
distributed over disjoint paths. Thus, we spread subscription messages
(not publications) over the spanning tree built to construct the
virtual ring. In the spanning tree layer subscription messages are
distributed through broadcasts. For maintenance of $\fwdTableWOUTv$
messages of the form \mbox{\textsc{Sub}$\langle r, C_S, P \rangle$}
are distributed with period~$\delta_S$. The spanning tree layer
provides an interface \mbox{\sendsub[Message msg]}, which is used by
to the pub/sub layer to send \textsc{Sub} messages. Hence, the virtual
ring layer is bypassed. The spanning tree layer acts as a filter for
the broadcasts. A subscription message sent by a node in the tree is
received by parent and child nodes only. Physically it can be received
by other nodes too, but these disregard such messages.

To avoid multiple delivery of \textsc{Sub} messages, they contain the
previous sender $r$ of the message, initially $r=\bot$. $C_S$ contains
the identifiers of the subscribed channels and $P$ all positions of
subscriber $s$, i.e., $P = Pos(s)$. Distributing the set of channels a
node has subscribed to in one message, instead of sending one message
per channel (as in~\cite{Siegemund_VR:2015}), reduces the number of
sent \textsc{Sub} messages by a factor of approximately $n_c$. This reduces the network load and thus, the possibility
for message collisions.

If a \textsc{Sub} message from a subscriber~$s$, forwarded by a
node~$u$, is received by a node~$v$, then $\fwdTable$ is updated using
\mbox{$UpdSn(c, SP)$}: If there exists a position $p_i \in P$, which
is ccw closer than the currently stored next subscriber $ns$ values in
$\fwdTable$, then it is replaced by $p_i$ for a given channel $c$
(details in Algorithm~\ref{alg:unsub}). E.g., if $Pos(v) = \langle
5,12,18\rangle$, $\fwdTable=\langle 14,14,20 \rangle$, and the new
subscriber positions are $Pos(s) = \langle 3,7\rangle$, then the
updated routing structure is \mbox{$\fwdTable = \langle 7,14,20
  \rangle$}, because position~7 is closer to position~5 than~14.

Before forwarding a message from a node~$u$ the parameter~$r$ is
altered by the forwarding node~$v$, i.e., $r := u$. If a node~$w$
receives a message with \mbox{$r=w$}, then $w$ discards the message.
This ensures that a node does not resend a previously send
\textsc{Sub} message. Leaves of the tree and subscribers do not
forward messages, they only update their routing structure
$\fwdTable$. Algorithm~\ref{alg:subscribeNew} describes the handling
of subscription messages. 

Figure~\ref{fig:exmpSubRingFULL} shows an example, which is kept
simple to increase the lucidity. It shows the subscription
distribution for a single channel in a line topology. The according
virtual ring which does not have any shortcuts is depicted as well.
When node~$a$ subscribes for the first time, node~$c$ already is a
subscriber. Node~$a$ broadcasts the initial \textsc{Sub}$\langle \bot,
\langle c\rangle,\langle 3,9\rangle \rangle$ message. Node~$c$ does
not forward it because it is a subscriber itself. Node~$d$ and~$b$
forward the subscription and change the variable~$r$ accordingly. As a
leaf, node~$f$ updates its routing structure but does not forward the
message. The changes of $\fwdTableWOUTv$ induced by the subscription
of node~$a$ are depicted in \ref{fig:exmpSubRingFULL} (middle).
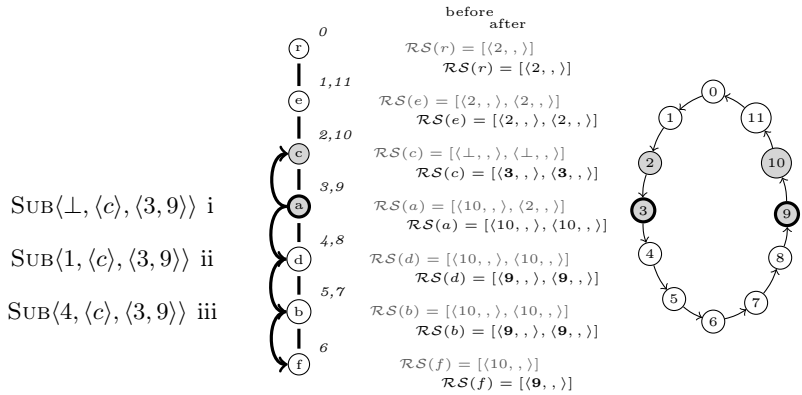
\begin{figure}[h]
\centering
\begin{tikzpicture}
\tikzset{
    mynode/.style={rectangle,align=center, rounded corners,draw=black, top color=white, bottom color=gray!45!black!30,	inner sep=0.2em, minimum size=1.6em, minimum width=9em, text centered},
    emptynode/.style={rectangle,align=center, rounded corners,draw=white, top color=white, bottom color=white,
    					thin, inner sep=0.2em, minimum size=1.6em, minimum width=9em, text centered},
    myarrow/.style={<->, >=latex', shorten >=1pt, thick},
    mylabel/.style={text width=7em, text centered},
    dotA/.style={circle, draw=black, font=\tiny, inner sep=0.5mm}, %
    pub/.style={circle, draw=black, fill=black!60!gray, text = white, font=\tiny, inner sep=0.5mm}, %
    subT/.style={circle, draw=black, fill=black!10!gray!30, font=\tiny, inner sep=0.5mm}, %
	subO/.style={circle, draw=black, very thick, fill=black!10!gray!30, font=\tiny, inner sep=0.5mm}, %
    edgeA/.style={shorten >= 0.6mm, shorten <= 0.6mm, dashed}, %
    edgeChosen/.style={shorten >= 0.6mm, shorten <= 0.6mm, very thick}, %
    edgeTree/.style={shorten >= 0.6mm, shorten <= 0.6mm, thin}, %
    dotVR/.style={circle, draw=black, font=\tiny, inner sep=0.5mm}, %
    edgeVR/.style={shorten >= 0.2mm, shorten <= 0.2mm, bend left},%
    Old/.style={black!30!gray},%
    connecto/.style={shorten >= 0.5mm, shorten <= 0.5mm}%

}  
\node[dotA, label={[font=\tiny]20:\textit{0}}] (d0) at (-1.4684,10.7017) {r};
\node[dotA, label={[font=\tiny]20:\textit{1,11}}] (d5) at (-1.4684,10.0017) {e};
\node[subT, label={[font=\tiny]20:\textit{2,10}}] (d3) at (-1.4684,9.3017) {c};
\node[subO, label={[font=\tiny]20:\textit{3,9}}] (d1) at (-1.4684,8.6017) {a};
\node[dotA, label={[font=\tiny]20:\textit{4,8}}] (d4) at (-1.4684,7.9017) {d};
\node[dotA, label={[font=\tiny]20:\textit{5,7}}] (d2) at (-1.4684,7.2017) {b};
\node[dotA, label={[font=\tiny]20:\textit{6}}] (d7) at (-1.4684,6.5017) {f};

\draw[edgeChosen] (d0) -- (d5);
\draw[edgeChosen] (d5) -- (d3);
\draw[edgeChosen] (d3) -- (d1);
\draw[edgeChosen] (d1) -- (d4);
\draw[edgeChosen] (d4) -- (d2);
\draw[edgeChosen] (d2) -- (d7);

\path[->] (d1)  edge [very thick, bend left = 90] (d3);
\path[->] (d1)  edge [very thick, bend right = 90] (d4);
\path[->] (d4)  edge [very thick, bend right = 90] (d2);
\path[->] (d2)  edge [very thick, bend right = 90] (d7);

\node[font=\footnotesize] (t0) at (-3.95,8.59) {\textsc{Sub$\langle\bot,
\langle c\rangle,
\langle 3,9\rangle\rangle$}$_{}$ i};
\node[font=\footnotesize] (t0) at (-3.95,7.89) {\textsc{Sub$\langle 1,
\langle c \rangle,
\langle 3,9 \rangle \rangle $}$_{}$ ii};
\node[font=\footnotesize] (t0) at (-3.95,7.19) {\textsc{Sub$\langle 4,
\langle c \rangle,
\langle 3,9 \rangle \rangle$}$_{}$ iii};

\tikzstyle{convex1}=[color=gray!30,];
\node[dotVR, label={[font=\tiny]90:$ $}] (v0) at (4.04,10.13) {0};
\node[dotVR,label={[font=\tiny]180:$ $}] (v1) at (3.47,9.78) {1};
\node[subT, label={[font=\tiny]179:$ $}] (v2) at (3.2,9.18) {2};
\node[subO, label={[font=\tiny]180:$ $}] (v3) at (3.11,8.55) {3};
\node[dotVR, label={[font=\tiny]180:$ $}] (v4) at (3.2,7.97) {4};
\node[dotVR, label={[font=\tiny]180:$ $}] (v5) at (3.52,7.37) {5};
\node[dotVR, label={[font=\tiny]270:$ $}] (v6) at (4.04,7.07) {6};

\node[Old] at (0.8,10.7) {\tiny $\mathcal{RS}(r)=[\langle2,,\rangle]$};
\node[Old] at (0.8,10) {\tiny $\mathcal{RS}(e)=[\langle2,,\rangle,\langle2,,\rangle]$};
\node[Old] at (0.8,9.3) {\tiny $\mathcal{RS}(c)=[\langle\bot,,\rangle,\langle\bot,,\rangle]$};
\node[Old] at (0.8,8.6) {\tiny $\mathcal{RS}(a)=[\langle10,,\rangle,\langle2,,\rangle]$};
\node[Old] at (0.8,7.9) {\tiny $\mathcal{RS}(d)=[\langle10,,\rangle,\langle10,,\rangle]$};
\node[Old] at (0.8,7.2) {\tiny $\mathcal{RS}(b)=[\langle10,,\rangle,\langle10,,\rangle]$};
\node[Old] at (0.8,6.5) {\tiny $\mathcal{RS}(f)=[\langle10,,\rangle]$};

\node[] at (1.3,10.44) {\tiny $\mathcal{RS}(r)=[\langle2,,\rangle]$};
\node[] at (1.3,9.74) {\tiny $\mathcal{RS}(e)=[\langle2,,\rangle,\langle2,,\rangle]$};
\node[] at (1.3,9.04) {\tiny $\mathcal{RS}(c)=[\langle\textbf{3},,\rangle,\langle\textbf{3},,\rangle]$};
\node[] at (1.3,8.34) {\tiny $\mathcal{RS}(a)=[\langle10,,\rangle,\langle10,,\rangle]$};
\node[] at (1.3,7.64) {\tiny $\mathcal{RS}(d)=[\langle\textbf{9},,\rangle,\langle\textbf{9},,\rangle]$};
\node[] at (1.3,6.94) {\tiny $\mathcal{RS}(b)=[\langle\textbf{9},,\rangle,\langle\textbf{9},,\rangle]$};
\node[] at (1.3,6.24) {\tiny $\mathcal{RS}(f)=[\langle\textbf{9},,\rangle]$};

\node[dotVR] (v7) at (4.61,7.32) {7};
\node[dotVR] (v8) at (4.93,7.92) {8};
\node[subO] (v9) at (5.02,8.5) {9};
\node[subT] (v10) at (4.88,9.18) {10};
\node[dotVR] (v11) at (4.61,9.78) {11};

\node[] (empty) at (3.62,6.3) {};

\path[->] (v0)  edge [bend right = 10] (v1);
\path[->] (v1)  edge [bend right = 10] (v2);
\path[->] (v2)  edge [bend right = 10] (v3);
\path[->] (v3)  edge [bend right = 10] (v4);
\path[->] (v4)  edge [bend right = 10] (v5);
\path[->] (v5)  edge [bend right = 10] (v6);
\path[->] (v6)  edge [bend right = 10] (v7);
\path[->] (v7)  edge [bend right = 10] (v8);
\path[->] (v8)  edge [bend right = 10] (v9);
\path[->] (v9)  edge [bend right = 10] (v10);
\path[->] (v10)  edge [bend right = 10] (v11);
\path[->] (v11)  edge [bend right = 10] (v0);

\node[] at (0.8,11.2) {\tiny before};
\node[] at (1.31,11.04) {\tiny after};




\end{tikzpicture}
\caption{Node~$a$ with positions 3 and 9 subscribes for the first time to channel $c$.}
\label{fig:exmpSubRingFULL}
\end{figure}

\newcommand{\reqRenewalC}{reqRenewalC}

\renewcommand\topfraction{0.85} 
\renewcommand\bottomfraction{0.85} 
\setlength\multicolsep{0pt}
\algrenewcommand\algorithmicindent{0.8em}%
\begin{algorithm}
\caption{Subscribing -- pub/sub Layer}
\label{alg:subscribeNew}
\begin{tabular*}{\textwidth}{@{}l@{~}l@{~}l}
Constants:& $\deltaS$ & resubscribe period (leasing period) \\
Variables:& $C_S$  & set of subscribed channels\\ 
& $\reqRenewalC$  & set of channels to be broadcasted\\
Functions:& $UpdSn(c, SP)$ & updates table $\fwdTable$ with positions P\\
Spanning tree layer API:& \sendsub[$\textsc{Msg}$] ~& broadcasts message $\textsc{Msg}$\\
& \textit{numChildren()} ~& returns number children in the tree\\
\end{tabular*}
\hrule
\begin{multicols}{2}
\begin{algorithmic}  
	\Function{\subscribe[c]}{}
		\If{($c \not \in C_S$)}
		\State{$C_S.add(c)$}	 
		\State \textsc{timer\_sub}.set(0)
		\EndIf
	\EndFunction
	\State
	\State Expiration of timer \textsc{timer\_sub}: %
	\State \textsc{timer\_sub}.set($\deltaS$)
	\State \sendsub[\textsc{Sub}$\langle \bot,C_S,P \rangle$]
	\State Upon $v$'s reception of $\textsc{Sub} \langle r,C,P\rangle$ from $u$
	\If{($r = v$)} 
	\State \Return
	\EndIf
	
	\ForAll{$c \in C_S$}
		\State $UpdSn(c, SP)$;
	\EndFor
	\State $C := C \setminus C_{S}$;
	\If{($ C \neq \emptyset ~ \wedge$ \textit{numChildren()} > 0)}
	\State{
		\sendsub[$\textsc{Sub} \langle u, C, P \rangle$]  
	}
	\EndIf
\end{algorithmic}
\end{multicols}
\end{algorithm}

\subsection{Publications}
\label{sec:pubParallel}



Publication messages need to be routed to subscribers only. Hence,
shortcuts can be used to skip non-subscribing nodes on the virtual
ring. Furthermore, publications of nodes with multiple positions can
be distributed concurrently over different paths. The following
propositions are tied to the fact that the virtual ring is built upon
a tree. Under a different scheme the routing still works, but some
properties, e.g., that each subscriber receives a publication only
once, are not guaranteed anymore.


\subsubsection*{Concurrent Routing.}
To explain publication routing on the virtual ring and the faced
challenges when routing messages concurrently we recap tree-based
routing. Each node maintains a routing table to identify branches
where at least one subscriber is present. A publisher distributes
messages into all such branches concurrently. The same reasoning is
conducted by forwarding nodes, while avoiding to send messages back to
previous senders. Trees are cycle free, hence, a publication is
delivered once per subscriber. In the virtual ring, shortcuts
introduce cycles. To avoid message duplication the concept of
\textit{routing into a branch} is transferred to the virtual ring.
Therefore, the \textit{end of a branch} is defined.



Nodes have multiple positions on the virtual ring, one for each
neighbor in the tree. Hence, sending a message from every position in
\mbox{$Pos(v) =\langle p_1, \dots, p_s\rangle$} to \mbox{$p_1+1,
  \dots, p_s+1$}, respectively is the equivalent of a tree node
sending into all branches. In the routing structure $\fwdTableWOUTv$
the next subscriber for each position is stored, this reflects a
node's understanding that a subscriber exists in a certain tree
branch. Therefore, if a publisher knows that there is at least one
subscriber in an interval $\mathcal{I} = [p_i, p_{i+1})$ for a given
channel $c$
then it sends a publication to a \textit{goal} position in
$\mathcal{I}$. The \textit{goal} position is the ccw closest one-hop
reachable position to the next subscriber in $\mathcal{I}$, i.e.,
$goal$ is either the next position on the virtual ring or a position
reachable by a shortcut.

Received publications are delivered to all nodes subscribing to the
message's channel. Regardless of the delivery, publications are
forwarded to ensure that all subscribers receive it. Forwarding of
publications is restricted to the interval they are sent into. To
avoid sending messages beyond interval borders the endpoint
$\borderPos$ of each $\mathcal{I}$ is attached to publication
messages: \mbox{\textsc{Pub}$\langle goal,ep,c,data \rangle$.} Where
$\borderPos$ is the right endpoint of \mbox{$\mathcal{I} = [p_i,
  p_{i+1})$}, i.e., $\borderPos = p_{i+1}$. A message is neither
routed to $\borderPos$ nor to a position beyond it. Parameters $goal$
and $ep$ are updated at every forwarding node. Parameter $data$
represents the payload.

The start position of an interval is the current position of a node
and the endpoint position is defined by the ccw next position of the
same node. Multiple delivery of a publication to nodes with multiple
positions in an interval is avoided as shown in
Lemma~\ref{lem:nesting}. 

\begin{lemma}
  The positions of nodes on the virtual ring are never interlaced.
  That is, a node~$v$ may have a position on the virtual ring which is
  followed by a node~$w$'s position, once another position of~$v$
  appears there cannot be a further position of~$w$.
\label{lem:nesting}
\end{lemma}\vspace*{-2mm}
\begin{proof}
The virtual ring is derived from a tree. A node has multiple positions if and only if it has children in the tree. All positions of a child branch are therefore nested in between two of its parents positions.
\end{proof}

As Lemma~\ref{lem:nesting} suggests, within a nodes's interval
$\mathcal{I}$ may be further intervals of other nodes. For the routing
this means, that a node forwarding a publication applies the same
reasoning as a publisher to determine how to forward messages. In the
tree this corresponds to branching. Each branch containing a
subscriber leads to an additional message sent concurrently. The
analog in the virtual ring is as follows: Each subscriber in the
interval \mbox{$\mathcal{I}_{f} = [p_i, p_{i+1})$} with $p_{i+1}$ ccw
in between $p_i$ and $ep$ forwards the \textsc{Pub} message. That is,
in the \textit{subsection} of the virtual ring bounded by the current
node position and the received endpoint position $ep$, independent
concurrent routing is conducted. Therefore, the parameters of the
\textsc{Pub} message are updated. The endpoint becomes $p_{i+1}$ if
$p_{i+1}$ is ccw between $p_i$ and $ep$ otherwise it stays unchanged.

Algorithm~\ref{alg:publishAugment} shows the handling of publications
and the calculation of associated endpoints. When a node generates a
publication with content $data$, then the \textit{handlePub()}
function is called, i.e., message \textsc{Pub}$\langle
P[0],P[0],c,data \rangle$ is sent.

\begin{theorem}
\label{tho:pubOnce}
In error-free phases subscribers receive \textsc{Pub} messages exactly once. 
\end{theorem}
\begin{proof}
  Once a position receives a publication message it is distributed
  over all possible positions with updated \textit{$\borderPos$}. This
  is equivalent to routing messages into branches of the underlying
  tree. Since parameter $\borderPos$ of a publication is closer or
  equal to the next position of the same node when the message is
  forwarded, it is assured that no further position of the same node
  receives a message again. For a particular position of a node,
  routing is conducted using a tree edge or a shortcut. A shortcut can
  only be used if $\fwdTable$ ensures that no subscriber is skipped.
  Hence, in the range of the tree between the position the shortcut
  leads to and the tree position which would be used instead
  (incremented current position) no subscriber
  exists.
\end{proof}

To illustrate the advantage of using shortcuts consider the topology
and the virtual tree graph in
Fig.~\ref{fig:exampleBackForth} and \ref{fig:VRrouting_homePos}. In
pure tree routing a message from node $c$ to $d$ is sent via node $e$.
With $\mathcal{PSVR}$ the direct shortcut between node~$c$ and~$d$ is
taken. The table in Fig.~\ref{fig:examplePos} shows the \textit{next
  subscriber} and the $goal$ positions. The next
subscriber is the according entry in $\fwdTableWOUTv{}$ for the stated
position. The publisher initiates two  delivery paths, one
for each position, i.e., for each interval. In the virtual ring in
Fig.~\ref{fig:VRrouting_homePos} these subsections are depicted as
light gray areas. One subsection starts at position~6 the other at~8
while $ep$ is the start position of the next subsection, respectively.
Publisher $d$ sends messages \mbox{\textsc{Pub}$\langle 7,8,c,data
  \rangle$} from position 6 and \mbox{\textsc{Pub}$\langle 2,6,c,data
  \rangle$} from position 8. Position~2 forwards the publication in
one interval with the borders $[2,4)$ with the message
\mbox{\textsc{Pub}$\langle 3,4,c,data \rangle$}. In
Fig.~\ref{fig:VRrouting_homePos} this is represented by the dark gray
area. In the interval $[4,6)$ no subscriber exists, hence, no message
is sent into the respective subsection. 

Figure~\ref{fig:expPub} shows an execution of
Algorithm~\ref{alg:publishAugment}. A virtual ring with shortcuts is
depicted. Furthermore, the table next to the ring shows the routing
table $\fwdTable$ for each node and a single channel, with each of its
own positions ($Pos$). A table cell represents one node, e.g., the
first node has the positions 0, 8, 12, and 14. Publishers are depicted
as black circles, and subscribers as gray ones.

\subsubsection*{Detailed Example for Parallel Publication Routing}

In Fig.~\ref{fig:expPub} two publishers exist, one at positions 1, 3, and 7 (referred to as node~$a$) and at position 15 and 19 (node~$b$). Both publishers send one message for each interval $[p_i, p_{i+1})$ a subscriber is present (see schedule for \circled{a} and \circled{b}).
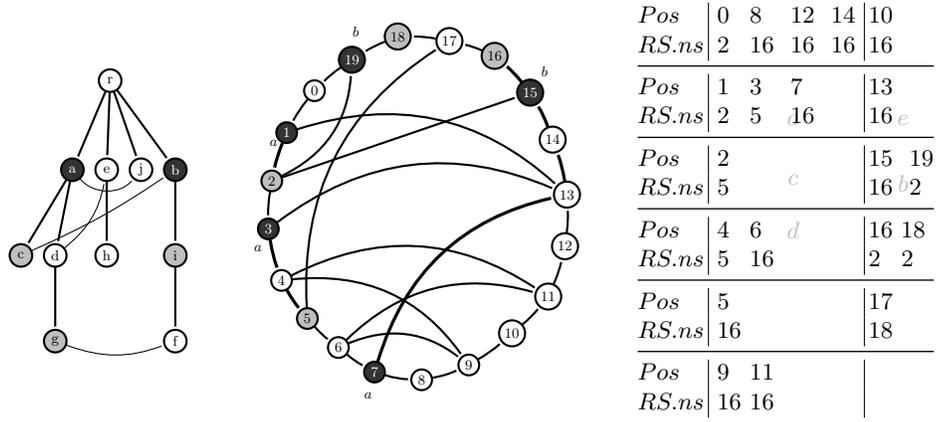
\begin{figure}[h]	
\begin{minipage}[c]{.25\linewidth}
\begin{tikzpicture}[
	scale=0.76, every node/.style={scale=0.76},
	-,>=stealth',shorten >=1pt,auto,node distance=3cm,
  thick,main node/.style={circle,draw,  inner sep=2, minimum size=5mm}, scale=0.6, every node/.style={scale=0.6}, hidden/.style={circle}, 
  publisher node/.style={circle, draw, fill=black!60!gray, text=white, inner sep=2, minimum size=5mm}, 
  subscriber node/.style={circle,draw, fill=gray!50!, inner sep=2, minimum size=5mm}]

\node[main node] (17) at (2,-7.5) {f};
\node[subscriber node] (16) at (2,-5) {i};
\node[publisher node] (15) at (2,-2.5) {b};

\node[main node] (13) at (1,-2.5) {j};

\node[main node] (10) at (0,-5) {h};
\node[main node] (9) at (0,-2.5) {e};
			
\node[subscriber node] (5) at (-1.5,-7.5) {g};
\node[main node] (4) at (-1.5,-5) {d};
\node[subscriber node] (2) at (-2.5,-5) {c};
\node[publisher node] (1) at (-1,-2.5) {a};
\node[main node] (0) at (0.1,0.1) {r};

  \path[every node/.style={font=\sffamily\small}]
  (0) edge[] node {} (1) edge[] node {} (9) edge[] node {} (13) edge[] node {} (15)
  (1) edge[] node {} (2)
  (2) edge[bend right = 5, thin] node {} (15)
  (1) edge[] node {} (4)
  (4) edge[] node {} (5)
  (4) edge[bend right = 20, thin] node {} (9)
  (5) edge[bend right = 20, thin] node {} (17)
  (9) edge[] node {} (10)
  (15) edge[] node {} (16)
  (16) edge[] node {} (17)
  (1) edge[bend right = 50, thin] node {} (13)
   	;
\end{tikzpicture}
\end{minipage}
\begin{minipage}[c]{.4\linewidth}
\begin{tikzpicture}[
	scale=0.76, every node/.style={scale=0.76},
	-,>=stealth',shorten >=1pt,auto,node distance=3cm,
  thick,main node/.style={circle,draw, inner sep=2}, scale=0.6, every node/.style={scale=0.6}, hidden/.style={circle}, publisher node/.style={circle, draw, fill=black!60!gray, text=white, inner sep=2}, subscriber node/.style={circle,draw, fill=gray!50!, inner sep=2}]

\node[publisher node] (19) at (-1.0186,4.4351) {19};
\node[subscriber node] (18) at (0.3181,5.1129) {18};
\node[main node] (17) at (1.8243,4.9811) {17};
\node[subscriber node] (16) at (3.1422,4.548) {16};
\node[publisher node] (15) at (4.1778,3.4749) {15};
\node[main node] (14) at (4.8344,2.104) {14};
\node[main node] (13) at (5.2512,0.5) {13};
\node[main node] (12) at (5.1944,-1.0061) {12};
\node[main node] (11) at (4.7049,-2.4746) {11};
\node[main node] (10) at (3.6506,-3.5478) {10};
\node[main node] (9) at (2.3891,-4.4703) {9};
\node[main node] (8) at (1.0147,-4.9034) {8};			
\node[publisher node] (7) at (-0.3597,-4.6774) {7};
\node[main node] (6) at (-1.414,-3.962) {6};
\node[subscriber node] (5) at (-2.3177,-3.1147) {5};
\node[main node] (4) at (-3.0708,-2.0039) {4};
\node[publisher node] (3) at (-3.4552,-0.5) {3};
\node[subscriber node] (2) at (-3.3552,0.9) {2};
\node[publisher node] (1) at (-2.9202,2.3076) {1};
\node[main node] (0) at (-2.1106,3.5314) {0};

  \path[every node/.style={font=\sffamily\small}]
  (0) edge[bend right=5] node {} (1)
   (1) edge[very thick,bend right=5] node {} (2)
   edge[bend left] node {} (13)    
    (2) edge[bend right=5] node {} (3)
        edge[] node {} (15)     
        edge[bend right] node {} (19)     
    (3) edge[bend right=5, very thick] node {} (4)
    edge[bend left] node {} (13)    
    (4) edge[bend right=5, very thick] node {} (5)
           edge[bend left] node {} (9)     
           edge[bend left] node {} (11)     
	(5) edge[bend right=5] node {} (6)
	edge[bend left] node {} (17)     
           
	(6) edge[bend right=5] node {} (7)
	edge[bend left] node {} (9)     
           edge[bend left] node {} (11)     
	(7) edge[bend right=7] node {} (8)
	edge[bend left, very thick] node {} (13)    
	(8) edge[bend right=7] node {} (9)
	(9) edge[bend right=5] node {} (10)
	(10) edge[bend right=5] node {} (11)
	(11) edge[bend right=5] node {} (12)
	(12) edge[bend right=5] node {} (13)
	(13) edge[bend right=5, very thick] node {} (14)
	(14) edge[bend right=5,very thick] node {} (15)
	(15) edge[bend right=5,very thick] node {} (16)
	(16) edge[bend right=5] node {} (17)
	(17) edge[bend right=5] node {} (18)
	(18) edge[bend right=5] node {} (19)
	(19) edge[bend right=5] node {} (0)
	;

\node[hidden] (ia) at (0.8,-4.2) {};
\node[hidden] (ib) at (-2.4088,-2) {};
\node[hidden] (iia) at (4.288,-2.1) {};
\node[hidden] (iib) at (3.2,-3.3168) {};
\node[hidden] (iiia) at (2,4.2) {};
\node[hidden] (iiib) at (3.1,3.7) {};
\node[hidden] (iva) at (-1.3336,3.5) {};
\node[hidden] (ivb) at (0,4.1) {};

\node[hidden] (1) at (-3.3,2) {$a$};
\node[hidden] (1) at (-3.7584,-1.0584) {$a$};
\node[hidden] (1) at (-0.5456,-5.3712) {$a$};

\node[hidden] (1) at (4.6,4.1) {$b$};
\node[hidden] (1) at (-0.9,5.2416) {$b$};
    ;


\tikzstyle{convex1}=[color=gray!15,];
\tikzstyle{convex2}=[color=gray!25,];
\tikzstyle{convex3}=[color=gray!45,];
\tikzstyle{convex4}=[color=gray!55,];

\end{tikzpicture}
\end{minipage}
\begin{minipage}[c]{.30\linewidth}
\begin{tikzpicture}
    \matrix[ampersand replacement=\&] {
        \node (species1) [font= \footnotesize] {
        \begin{tabular}{@{}l | @{~}l@{}l@{~}l@{~}l@{~} | l@{~}l@{}}
         $Pos$&$0$&~$8$&~$12$&~$14$&$10$&~\\ 
        $\fwdTableWOUTv.ns$&$2$&~$16$&~$16$&~$16$&$16$&~\\
        \midrule
        $Pos$&$1$&~$3$&~$7$ & &$13$&~\\ 
        $\fwdTableWOUTv.ns$&$2$&~$5$&~$16$& &$16$&~\\
        \midrule
        $Pos$&$2$&~&~& &$15$&~$19$\\ 
        $\fwdTableWOUTv.ns$&$5$&~&~& &$16$&~$2$\\
        \midrule
        $Pos$&$4$&~$6$&~&  &$16$&$18$\\ 
        $\fwdTableWOUTv.ns$ &$5$&~$16$&~& &$2$&$2$\\
        \midrule
        $Pos$&$5$ &~&~& &$17$&\\ 
        $\fwdTableWOUTv.ns$ &$16$&~&~& &$18$&\\
        \midrule
        $Pos$&$9$&~$11$&~& & &\\ 
        $\fwdTableWOUTv.ns$ &$16$&~$16$&~& & &\\
        \bottomrule
        \end{tabular}

        };
        \& 
        \\        
};
\begin{pgfonlayer}{background}
\node[text=gray!50] at (0.1,1.25)  {$a$};
\node[text=gray!50] at (1.56,0.4)  {$b$};
\node[text=gray!50] at (0.1,0.45)  {$c$};
\node[text=gray!50] at (0.1,-0.22)  {$d$};
\node[text=gray!50] at (1.56,1.25)  {$e$};
\end{pgfonlayer}
\end{tikzpicture}


\ra{1}
\end{minipage}
\caption{Publication routing example on virtual ring. Black and gray positions are publishers and subscribers, respectively.}
\label{fig:expPub}
\end{figure}

\begin{center}
\begin{minipage}[t]{.48\linewidth}
\centering
\ra{1.1}
	\begin{tabular}{@{}l@{}clll@{}}
	& from &  & to & $\borderPos$\\
	\cline{2-5}
\circled{$a$} & 	1 & -> & 2 & 3\\
 & 	3 & -> & 4 & 7\\
& 	7 & -> & 13 & 1\\
	\cline{2-5}
	\end{tabular}
\end{minipage}
\begin{minipage}[t]{.48\linewidth}
\centering
\ra{1.1}
	\begin{tabular}{@{}l@{}clll@{}}
&	from &  & to & $\borderPos$\\
	\cline{2-5}
\circled{$b$}&	15 & -> & 16 & 19\\
 &	19 & -> & 2  & 15\\
&&&&\\
	\cline{2-5}
	\end{tabular}
\end{minipage}
\end{center} 
		
When a subscriber receives a publication, it delivers the message, then it evaluates if the message has to be forwarded. If one or more subscribers exist within the received $\borderPos$, and if the calculated goal position does not lie beyond any of its other positions or the received $\borderPos$, then a new $\borderPos$ is calculated and the message is altered before it is sent.

We focus on the publication from node~$a$. Positions 2, 4, and 13 received the publication and forward the message according to the schedule for \circled{$c$} \circled{\small$d$} \circled{$e$}.

\begin{center}
\begin{minipage}[t]{.32\linewidth}
\centering
\ra{1.1}
	\begin{tabular}{@{}l@{}c@{}l@{}l@{}r@{}}
&	from &  & to &~$\borderPos$\\
	\cline{2-5}
\circled{$c$}&	2 & not &  & \\
&&&&\\
	\cline{2-5}
	\end{tabular}
\end{minipage}
\begin{minipage}[t]{.32\linewidth}
\centering
\ra{1.1}
	\begin{tabular}{@{}l@{}c@{}l@{}l@{}r@{}}
&	from &  & to&~$\borderPos$\\
	\cline{2-5}
\circled{\small$d$}&	4 & -> & 5 & 6\\
 &	6 & not &   & \\
	\cline{2-5}
	\end{tabular}
\end{minipage}
\begin{minipage}[t]{.32\linewidth}
\centering
\ra{1.1}
	\begin{tabular}{@{}l@{}c@{}l@{}l@{}r@{}}
&	from &  & to &~$\borderPos$\\
	\cline{2-5}
\circled{$e$}&	13 & -> & 14 & 1\\
 &	&  &  & \\
	\cline{2-5}
	\end{tabular}
\end{minipage}
\end{center} 

Position 2 does not forward the message since the $\borderPos$ is~3. Position~4 sends a message to 5 which is not forwarded by~5 because the $\borderPos$ was changed from 7 to 6 which is another position from the node at position~4. Finally, the message from position~13 is forwarded.

At position~14 only one message is sent, the one with destination position 15, because the $\borderPos$ is still~1 and positions~8, 12, and~14 (all belong to the same node) have the same next subscriber position. Nevertheless, the $\borderPos$ is changed at position 14 to 0. Position 15 forwards to 16. 

At 16 the message is delivered. Position 18, which is the second position of the node at position 16, does not forward the publication, because the next subscriber, position 2, is beyond the current $\borderPos$ (0). Therefore, the publication is not forwarded any further, which is desired since all subscribers got the publication.

In this example seven messages are sent to deliver the publication. With the algorithm in~\cite{Siegemund_VR:2015} twelve messages are necessary. With \PSVR positions 8 to 12, and position 17 are skipped.

\subsubsection*{Resolving drawbacks of~\cite{Siegemund_VR:2015}.} In
the related work section two shortcomings of~\cite{Siegemund_VR:2015}
concerning publications were mentioned. Firstly, \textit{nodes
  receive publications multiple times}. For the example in
Fig.~\ref{fig:Everything} this means that the path a \textsc{Pub}
message travels, starting at position~6 is: $\langle
(6),7,8,2,3,4\rangle$, i.e., $\langle (d),b,d,c,a,c\rangle$. Node~d
(positions 6 and 8) receives its previously published message in order
to forward it. Additionally, node~c receives the same publication
twice. As can be examined in Fig.~\ref{fig:VRrouting_homePos}, with
$\mathcal{PSVR}$ two messages travel: $\langle (6),7\rangle\langle
(8),2,3\rangle$, i.e., $\langle (d),b\rangle\langle (d),c,a\rangle$.
This is a considerable improvement.

\begin{figure}
\begin{subfigure}[b]{0.17\linewidth}
\centering
	\begin{tikzpicture}
\tikzset{
    mynode/.style={rectangle,align=center, rounded corners,draw=black, top color=white, bottom color=gray!45!black!30,	inner sep=0.2em, minimum size=1.6em, minimum width=9em, text centered},
    emptynode/.style={rectangle,align=center, rounded corners,draw=white, top color=white, bottom color=white,
    					thin, inner sep=0.2em, minimum size=1.6em, minimum width=9em, text centered},
    myarrow/.style={<->, >=latex', shorten >=1pt, thick},
    mylabel/.style={text width=7em, text centered},
    dotA/.style={circle, draw=black, font=\tiny, inner sep=0.5mm}, %
    pub/.style={circle, draw=black, fill=black!60!gray, text = white, font=\tiny, inner sep=0.5mm}, %
    sub/.style={circle, draw=black, fill=black!10!gray!30, font=\tiny, inner sep=0.5mm}, %
    edgeA/.style={shorten >= 0.6mm, shorten <= 0.6mm, dashed}, %
    edgeChosen/.style={shorten >= 0.6mm, shorten <= 0.6mm, very thick}, %
    edgeTree/.style={shorten >= 0.6mm, shorten <= 0.6mm, thin}, %
    dotVR/.style={circle, draw=black, font=\tiny, inner sep=0.5mm}, %
    edgeVR/.style={shorten >= 0.2mm, shorten <= 0.2mm, bend left},%
    connecto/.style={shorten >= 0.5mm, shorten <= 0.5mm}%

}  
\node[dotA, label={[font=\tiny]20:\textit{0}}] (d0) at (-0.5,10) {r};
\node[sub, label={[font=\tiny]20:\textit{3}}] (d1) at (-1,7) {a};
\node[dotA, label={[font=\tiny]20:\textit{1,5,9}}] (d5) at (-0.5,9) {e};
\node[sub, label={[font=\tiny]130:\textit{2,4}}] (d3) at (-1,8) {c};
\node[sub, label={[font=\tiny]20:\textit{7}}] (d2) at (0,7) {b};
\node[pub, label={[font=\tiny]20:\textit{6,8}}] (d4) at (0,8) {d};

\draw[edgeChosen] (d0) -- (d5);
\draw[edgeChosen] (d5) -- (d3);
\draw[edgeChosen] (d5) -- (d4);
\draw[edgeTree] (d4) -- (d3);
\draw[edgeChosen] (d4) -- (d2);
\draw[edgeChosen] (d1) -- (d3);
\draw[edgeTree] (d1) -- (d2);
\end{tikzpicture}
	\caption{Graph}
	\label{fig:exampleBackForth}
\end{subfigure}
\begin{subfigure}[b]{0.35\linewidth}
\centering
	\usetikzlibrary{backgrounds}

\begin{tikzpicture}
\tikzset{
    mynode/.style={rectangle,align=center, rounded corners,draw=black, top color=white, bottom color=gray!45!black!30,	inner sep=0.2em, minimum size=1.6em, minimum width=9em, text centered},
    emptynode/.style={rectangle,align=center, rounded corners,draw=white, top color=white, bottom color=white,
    					thin, inner sep=0.2em, minimum size=1.6em, minimum width=9em, text centered},
    myarrow/.style={<->, >=latex', shorten >=1pt, thick},
    mylabel/.style={text width=7em, text centered},
    dotA/.style={circle, draw=black, font=\tiny, inner sep=0.5mm}, %
    pub/.style={circle, draw=black, fill=black!60!gray, text = white, font=\tiny, inner sep=0.5mm}, %
    sub/.style={circle, draw=black, fill=black!20!gray!30, font=\tiny, inner sep=0.5mm}, %
    edgeA/.style={shorten >= 0.6mm, shorten <= 0.6mm, dashed}, %
    edgeChosen/.style={shorten >= 0.6mm, shorten <= 0.6mm, very thick}, %
    edgeTree/.style={shorten >= 0.6mm, shorten <= 0.6mm, thin}, %
    dotVR/.style={circle, draw=black, fill=white, font=\tiny, inner sep=0.5mm}, %
    edgeVR/.style={shorten >= 0.2mm, shorten <= 0.2mm, bend left},%
    connecto/.style={shorten >= 0.5mm, shorten <= 0.5mm}%

}  
\node[dotVR] at (3,		10) (v0) {0};
\node[dotVR] at (2.3,	9.65) (v1) {1};
\node[sub, label={[font=\tiny]240:\textit{(ii) $ep$:4}}] at (2,		9) (v2) {2};
\node[sub] at (2,		8.2) (v3) {3};
\node[sub] at (2.3,	7.55) (v4) {4};
\node[dotVR] at (3,		7.2) (v5) {5};
\node[pub, label={[font=\tiny]30:\textit{(i) $ep$:8}}] at (3.7,	7.55) (v6) {6};
\node[sub] at (4,		8.2) (v7) {7};
\node[pub] at (4,		9) (v8) {8};
\node[dotVR] at (3.7,	9.65) (v9) {9};

\path[] (v0)  edge [bend right = 15] (v1);
\path[] (v1)  edge [bend right = 10] (v2);
\path[->] (v2)  edge [bend right = 10] (v3);
\path[] (v3)  edge [bend right = 10] (v4);
\path[] (v4)  edge [bend right = 15] (v5);
\path[] (v5)  edge [bend right = 15] (v6);
\path[->] (v6)  edge [bend right = 10] (v7);
\path[] (v7)  edge [bend right = 10] (v8);
\path[] (v8)  edge [bend right = 10] (v9);
\path[] (v9)  edge [bend right = 15] (v0);

\path[] (v3)  edge [] (v7);
\path[] (v2)  edge [bend left = 10,<-] node[font=\tiny,above]{\textit{(i) $ep$:6}} (v8) edge [bend left = 10] (v6);
\path[] (v4)  edge [bend left = 10] (v8) edge [bend right = 10] (v6);

\begin{scope}[on background layer]
\draw[gray!90, double=gray!10,double distance=4.8mm,smooth,line cap=round,tension=0.4] plot coordinates {(v6)(v7)};
\draw[gray!90, double=gray!10,double distance=4.8mm,smooth,line cap=round,tension=0.4] plot coordinates {(v8) (v9) (v0) (v1) (v2) (v3) (v4) (v5)};
\draw[gray!90, double=gray!65,double distance=4mm,smooth,line cap=round,tension=0.4] plot coordinates {(v2)(v3)};
\end{scope}

\end{tikzpicture}
	\caption{Virtual ring graph}
	\label{fig:VRrouting_homePos}
\end{subfigure}
    \begin{subfigure}[b]{0.4\columnwidth}
        \centering
		{\begin{tabular}{@{}l c c c cc cc ccc@{}}   	
		   	node				& r & a & b &\multicolumn{2}{c}{c} &\multicolumn{2}{c}{d} &\multicolumn{3}{c}{e}  \\
		    pos. 			& 0 & ~3~ & 7 & ~2 & 4~ & ~6 & 8~ & 1 & 5 & 9\\
		    next subscr.   & 2 & ~4~ & 2 & ~3 & 7~ & ~7 & 2~ & 2 & 7 & 2\\
		    goal    			& 1 & ~4~ & 8 & ~3 & 6~ & ~7 & 2~ & 2 & 6 & 0\\
		    \end{tabular} }
        \caption{Forwarding positions}\label{fig:examplePos}      
    \end{subfigure} 
\caption{Illustration of the forwarding process (Subscribers: gray;
  publisher: black)} 
   \label{fig:Everything}
\end{figure}
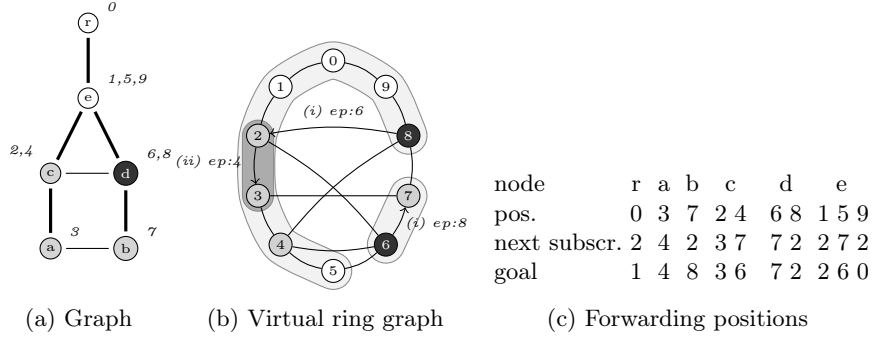

\begin{algorithm}[h]
\caption{Handling and forwarding of publications}
\label{alg:publishAugment}
{\scriptsize
\begin{tabular*}{\textwidth}{@{}ll@{}}
	\multicolumn{2}{@{}l}{\textbf{API provided by virtual ring layer (VR):} }\\
	~\getPosClosestTo[p, goal] 	  & \parbox[t]{8.2cm}{returns largest ccw position beyond \textit{p} and prior (or equal to) \textit{goal} within neighbor positions}\\
	~\sendRing[p, msg]& sends message $msg$ to position $p$\\
	~\isBetween[test, left, right]& \parbox[t]{8.2cm}{checks if position \textit{test} is in ccw ring segment bounded by positions \textit{left} and \textit{right}\\ note: \textit{isBetween}$(x,y,y)=$\textit{true} for arbitrary positions $x$ and $y$}\\
	~\textit{\textbf{deliver}}($data$) &  delivers the $data$ to the application	
\end{tabular*}		}
\hrule
\begin{multicols}{2}
\begin{algorithmic}
\algrenewcommand\algorithmicindent{0.5em}%
\Function{\publish[$c, data$]}{}
	\State{\textit{\textbf{handlePub}}($P[0],P[0],c,data$)}
\EndFunction	
\\
\begin{center}\rule[+1mm]{20mm}{.1pt}%
\end{center}
\State Upon reception of $\textsc{Pub} \langle curPos, ep, c, data\rangle$
\If{($c \in C_S$)}
	\State{\textit{\textbf{deliver}}($data$)}	
\EndIf
\State{\textit{\textbf{handlePub}($curPos, ep, c, data$)}}
\begin{center}\rule[+1mm]{20mm}{.1pt}%
\end{center}
\Function{\textit{\textbf{handlePub}($curPos, ep, c, data$)}}{}
\ForAll{$p \in P$}
\State{$nextS := \fwdTableWOUTv{}[indexOf(c)][indexOf(p)]$}
\State{$newEp := $\textit{\textbf{calcNewEP}}($p, ep$)}
\If{(\isBetween[$nextS, curPos, newEp$])} 
\State{$goal := $\getPosClosestTo[$p, nextS$]}
\State {\textit{\textbf{sendOnRing}}($goal$,  }
\State{~~\textsc{Pub}$\langle goal, newEp, c, data \rangle$)}
\EndIf
\EndFor
\EndFunction
\\
\Function{\textit{\textbf{calcNewEp}($p$, $maxEp$)}}{}
\State $i := indexOf(p)$
\State $epIndex := i+1 \bmod |P|$
\If{(\textit{\textbf{isBetween}}($P[epIndex], p, maxEp$))}
\State{\Return $P[epIndex]$}
\Else \State{\Return $maxEp$}
\EndIf
\EndFunction
\end{algorithmic}
\end{multicols}
\end{algorithm}


Secondly, \textit{publications travel further on the virtual ring as
  the last subscriber}. Consider an example where the next
subscriber~$p_w$ lies beyond a publisher~$p_v$, as depicted in
Fig.~\ref{fig:toFar}. In~\cite{Siegemund_VR:2015} a publication from
$p_v$ is forwarded until a node at position~$p_u$ can determine that
forwarding leads to routing the message past or to the original
publisher~$p_v$, then the node ceases forwarding. With
$\mathcal{PSVR}$, due to the definition of the end
position~$\borderPos$ and the knowledge of the next subscriber~$p_w$,
such a situation is recognized by the ccw \textit{last}
subscriber~$p_t$ before publisher~$p_v$. $p_t$ checks if the next
subscriber is between the current position and $\borderPos$. If this
is not the case, the message is not forwarded. Hence, in
Fig.~\ref{fig:toFar} four avoidable messages, starting at position~3
successively to position~7, are sent with~\cite{Siegemund_VR:2015}
compared to~$\mathcal{PSVR}$.

\begin{figure}[h]
\centering
\begin{tikzpicture}
\tikzset{
    mynode/.style={rectangle,align=center, rounded corners,draw=black, top color=white, bottom color=gray!45!black!30,	inner sep=0.2em, minimum size=1.6em, minimum width=9em, text centered},
    myarrow/.style={<->, >=latex', shorten >=1pt, thick},
    msgEdge/.style={-, >=latex', shorten >=3pt},
    mylabel/.style={text width=7em, text centered},
    dotA/.style={circle, draw=black, font=\tiny, inner sep=0.5mm}, %
    pub/.style={circle, draw=black, fill=black!60!gray, text = white, font=\tiny, inner sep=0.5mm}, %
    sub/.style={circle, draw=black, fill=gray!70, text = white, font=\tiny, inner sep=0.5mm}, %
    edgeA/.style={shorten >= 0.6mm, shorten <= 0.6mm, dashed}, %
    edgeChosen/.style={shorten >= 0.6mm, shorten <= 0.6mm, very thick}, %
    edgeTree/.style={shorten >= 0.6mm, shorten <= 0.6mm, thin}, %
    dotVR/.style={fill = white, inner sep=-2em, minimum size=3.8mm, circle, draw=black, font=\tiny, inner sep=0.5mm}, %
    edgeVR/.style={shorten >= 0.2mm, shorten <= 0.2mm, bend right=40, densely dotted},%
    connecto/.style={shorten >= 0.5mm, shorten <= 0.5mm}%

}

\node[] (v2) at (-0.6,7) {};
\node[sub, label={[font=\scriptsize]$p_t$}] (v3) at (-0.01,6.78) {3};
\node[dotA, label={[font=\tiny]180:$ $}] (v4) at (0.85,6.59) {4};
\node[dotA, label={[font=\tiny]180:$ $}] (v5) at (1.79,6.46) {5};
\node[dotA, label={[font=\tiny]270:$ $}] (v6) at (2.78,6.42) {6};
\node[dotA, label={[font=\scriptsize]$p_u$}] (v7) at (3.71,6.42) {7};
\node[dotA] (v8) at (4.7,6.47) {8};
\node[pub, label={[font=\scriptsize]$p_v$}] (v9) at (5.58,6.57) {9};
\node[sub,, label={[font=\scriptsize]$p_w$}] (v10) at (6.4,6.79) {10};
\node[] (v11) at (6.72,7.12) {};

\path[->] (v2)  edge [bend right = 5]  (v3);
\path[->] (v3)  edge [bend right = 5] (v4);
\path[->] (v4)  edge [bend right = 5] (v5);
\path[->] (v5)  edge [bend right = 5] (v6);
\path[->] (v6)  edge [bend right = 5] (v7);
\path[->] (v7)  edge [bend right = 5] (v8);
\path[->] (v8)  edge [bend right = 5] (v9);
\path[->] (v9)  edge [bend right = 5] (v10);
\path[-] (v10)  edge [bend right = 5] (v11);

\path[->] (v7)  edge [bend left = 35] node[above] {\scriptsize  too far!}  (v10);


\node (v1) at (4.7086,7.3307) {};
\node (v12) at (5.1368,6.8925) {};
\node (v13) at (5.12,7.34) {};
\node (v14) at (4.7,6.9) {};
\draw[red]  (v1) edge (v12);
\draw[red]  (v13) edge (v14);
\end{tikzpicture}
\caption{Virtual ring section. Shortcoming of~\cite{Siegemund_VR:2015}: unnecessary  forwarding}
\label{fig:toFar}
\end{figure}
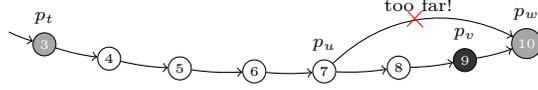

\subsection{Implicit Unsubscription Handling}
\label{sec:unSub}
A node $v$ that ends a subscription to a channel removes the
respective channel identifier from $C_S$. If $C_S$ becomes empty,
then~$v$ ceases to send \textsc{Sub} messages. This triggers updates
in the routing structure $\fwdTableWOUTv$ at other nodes. If a value
in $\fwdTableWOUTv$ has not been renewed after the leasing period
$\delta_S$, then it is identified as \textit{stale}. Stale entries are
used for routing nonetheless, i.e., for incoming publication
forwarding the staleness of the entry is irrelevant. When a \textsc{Sub} message with a ccw closer subscriber position is received, the stale value is replaced and time-stamp $\expireTimer$ is renewed.

If a stale value is not replaced in this way, then a temporary new
next subscriber $nstmp$ is stored when the next \textsc{Sub} message
is received. This $nstmp$ value is treated the same way as the
according stale value in $\fwdTableWOUTv$. Initially $nstmp := \bot$.
When a \textsc{Sub} message is received $nstmp$ is set to the closest
ccw subscribing position stated in the message. For every received
\textsc{Sub} message $nstmp$ is updated to store the closest ccw
subscriber position. While $nstmp$ is updated, routing for
\textsc{Pub} messages still refers to the stale value. After an update
period $T_{w\slash back}$ $nstmp$ replaces the stale value $ns$. The
time-stamp $\expireTimer$ is set to the current time and $nstmp$ resets
to $\bot$. Algorithm~\ref{alg:unsub} describes the details of the
already mentioned $UpdSn()$ function.

\setlength\multicolsep{0pt}
\algrenewcommand\algorithmicindent{0.7em}%

\begin{algorithm}
\caption{Unsubscriptions}
\label{alg:unsub}
\begin{tabular*}{\textwidth}{@{}lll@{}}
Constants: & $T_{clean}$ & clean timer expiration time \\
 		   & $T_{w\slash back}$ & write back period $\rightarrow$ temp value replaces stale value  \\
Functions: &  $isStale(t_{s}, exptimer)$ ~& return $currentLocalTime() - t_{s}$ > $exptimer$\\
\end{tabular*}
\hrule
\begin{multicols}{2}
\begin{algorithmic}  
\Function{$UpdSn(c, SP)$}{}
	\ForAll{$sp \in SP$}
		\ForAll{$rs_j \in RS[c]$}
			\If{(\textit{\textbf{isBetween}}($sp, P[j], rs_j.ns$))}				
				\State{$rs_j.ns := sp$}
				\State{$rs_j.\expireTimer := currentLocalTime()$}									 
			\ElsIf{($isStale(rs_j.t_s, \delta_S)$)}
				\If{(\textit{\textbf{isBetween}}($sp, P[j], rs_j.nstmp$))}
					\State{$rs_j.nstmp := sp$}						
				\EndIf	 
				\EndIf
		\EndFor
	\EndFor		
\EndFunction
\State{\scriptsize //\textsc{Timer\_Clean} initialized on system startup}
\State{Expiration of timer \textsc{Timer\_Clean}:}
\State{\textsc{Timer\_Clean}.set($T_{clean}$)}
\ForAll{$rs \in RS$}
	\If{($isStale(rs.t_s, T_{w\slash back})$)}
		\State{$rs.ns := rs.nstmp$}
		\State{$rs.\expireTimer := currentLocalTime()$}
		\State{$rs.nstmp := \bot$}
	\EndIf
\EndFor
\State{}

\end{algorithmic}
\end{multicols}
\end{algorithm}

Routing is correct during the whole process, that is, no subscriber is
skipped. When a node unsubscribes or an error in $\fwdTableWOUTv$
occurs, it takes at most $T_{w\slash back}$ periods of time until
$\fwdTableWOUTv$ is consistent again. The burden on memory for the
presented unsubscribing scheme is manageable. For each node position
the temp value and $\expireTimer$ has to be accounted for, typically
for each node that means $(2 + 4)C_N$ Bytes. Note that $\deltaS$ can
be  constant or determined during runtime, as
it has a strong correlation to the length of the virtual ring. When the virtual ring is constructed the root node sends a \textsc{Down} message including starting positions of each node into the ring~\cite{Siegemund_VR:2015}. The root node has knowledge of the tree and the ring size. Hence, attaching this value to the \textsc{Down} message of the virtual ring setup algorithm can be realized conveniently. The number of nodes, i.e., the length of the ring can then be used to calculate $\delta_S$. 

\subsection{Self-stabilizing Properties}
Self-stabilization is ensured by the leasing technique. Through the
periodic renewal of subscriptions routing tables are continually
updated and errors are fixed. Storing a time-stamp of the last update
$\expireTimer$ in the routing structure $\fwdTable$ ensures that stale
values can be recognized. Hence, inconsistencies due to message errors,
loss, or obstruction are corrected. Proper publication routing is
ensured by the correctness of $\fwdTable$. Unsubscribing is
self-stabilizing as well. To unsubscribe from a channel a node removes
the channel identifier from $C_S$, this ceases sending \textsc{Sub}
messages. The underlying structures, virtual ring and spanning tree
are built using self-stabilizing algorithms. They are tied together
using collateral composition where a layer does not influences a layer
below.

Self-stabilizing algorithms inherently can not locally decide if the
system is in a globally correct state. Thus, in a faulty case no
guarantees can be given, but that eventually the system will recover.
$\mathcal{PSVR}$ handles dynamic addition and removal of nodes, after
addition to the virtual ring and the dispatch of the first
\textsc{Sub} message it takes no longer than $O(n)$ rounds until
\textsc{Pub} messages will be received.

\section{Evaluation}

$\mathcal{PSVR}$ presents a compromise of size and maintenance effort
for routing tables and routing paths lengths. In order to assess the
increase of the path's lengths a comparison with two routing
strategies was done. In alternative $\mathcal{T}_D$ we computed a
breadth-first tree for each node and recursively pruned leaves not
corresponding to subscribers. Publications made by a node were
forwarded via the corresponding bfs-tree. Alternative $\mathcal{T}_S$
followed the common approach of a single routing tree. We chose a
bfs-tree rooted at a central node. The first alternative comes close
to the optimal structure, i.e., a Steiner tree. We analyzed connected
graphs $G(n,p)$ using the Erdős-Rényi model. The  message gain in percent is
calculated by $\sfrac{100B}{A} -100$, where $B$ is the number of
messages needed by $\mathcal{PSVR}$ and $A$ is the number of message
needed by the approach it is compared to.

The results indicate that
the difference between average path lengths decreases with increasing
density and with an increase of the number $s$ of subscribers. In
Fig.~\ref{fig:eval:AvsBvsC} the gain for both approaches compared to
$\mathcal{PSVR}$ is depicted. For example for $n\le 100$ and $s\ge
10$ the overhead of $\mathcal{PSVR}$ is less than 8 \%. The same trend
-- but at a lower level -- was observed for $\mathcal{T}_S$. We
conclude that except for very small numbers of subscribers the
overhead of $\mathcal{PSVR}$ with respect to path lengths is
surprisingly low. With increasing density the number of shortcuts
increases, allowing for shorter routing paths. Furthermore, with
growing number of nodes the gain follows the same distribution.

\begin{figure}[t]
	\centering
	\includegraphics[width=1\textwidth]{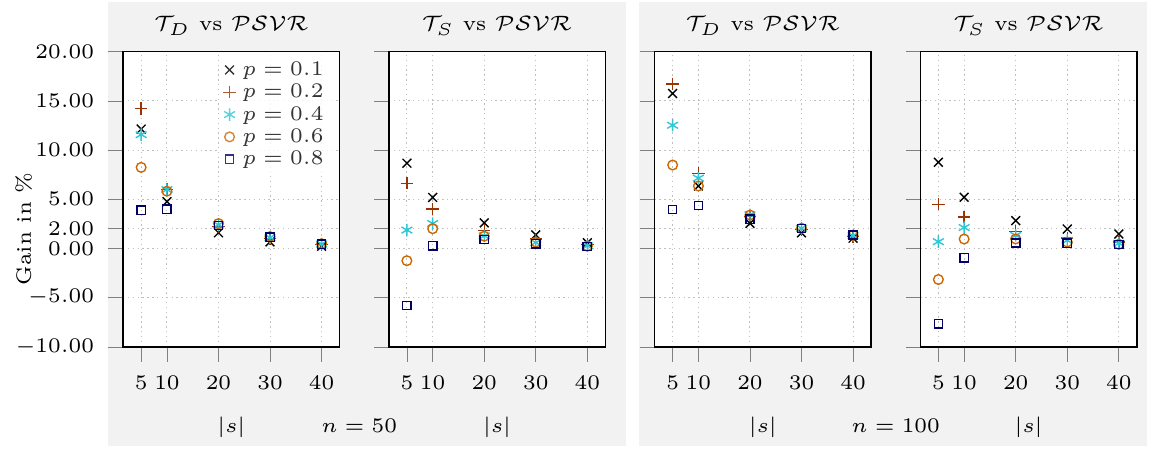}
	\captionof{figure}{Comparison of $\mathcal{T}_D$ and $\mathcal{T}_S$ vs $\mathcal{PSVR}$}
	\label{fig:eval:AvsBvsC}
\end{figure}

Next we analyzed the delivery ratio using implementations based on the
\OMNET simulation environment and the \MIXIM framework to employ a
radio model compared to the self-stabilizing tree based approach by
Shen et al.~\cite{Shen:2007}. Both approaches, Shen and $\mathcal{PSVR}$ use the same dynamically computed spanning tree. The throughput of both approaches is close to identical as presented in Fig.~\ref{fig:appendix:throughput}. Even though $\mathcal{PSVR}$ needs to maintain the virtual ring structure, shorter routes as depicted in~Fig.~\ref{fig:appendix:histogram}, compensate this handicap.

\vspace{-3mm}
\begin{figure}[h]
	\centering
	\includegraphics[width = \textwidth]{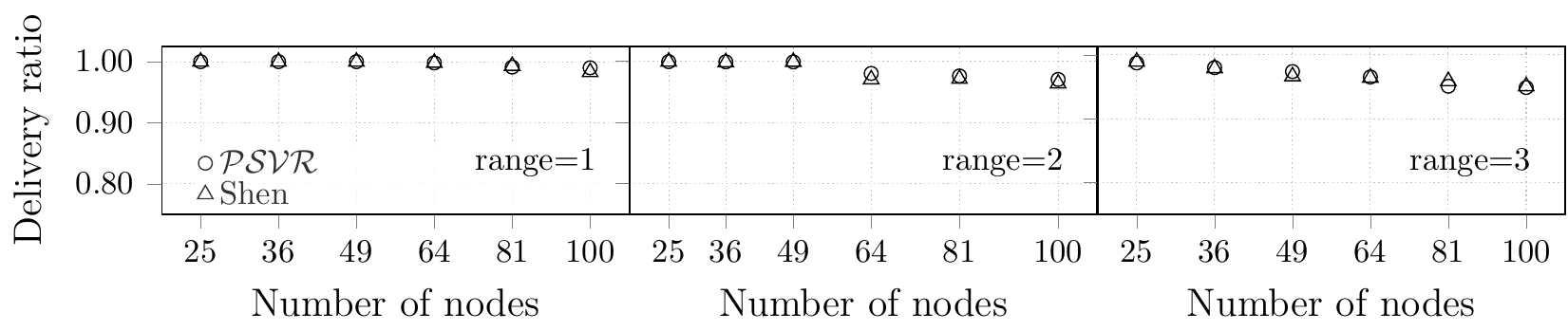}
	\caption{Publication delivery for varying densities, compared to Shen's approach.}
	\label{fig:appendix:throughput}
\end{figure}

\vspace*{-5mm}
Figure~\ref{fig:appendix:histogram} shows the constructed path lengths
for simulations involving a radio model and the complete network stack
described in Section~\ref{sec:foundation}. The depicted histogram
shows the route lengths for a scenario with a single subscriber and
each node is a publisher. As can be seen $\mathcal{PSVR}$ constructs
more short routes (i.e., up to 3 hops) as well as shorter
routes on average than Shen's algorithm. Increasing the number of
subscribers diminishes the gain of $\mathcal{PSVR}$ to the
point that all nodes are subscribers and no shortcut is taken anymore
but $\mathcal{PSVR}$ resembles routing on a spanning tree, i.e.,
$\mathcal{PSVR}$ falls back to the  Shen's approach.

\begin{figure}
\centering
\begin{minipage}{.47\textwidth}
 \centering
  \includegraphics[width=1\textwidth]{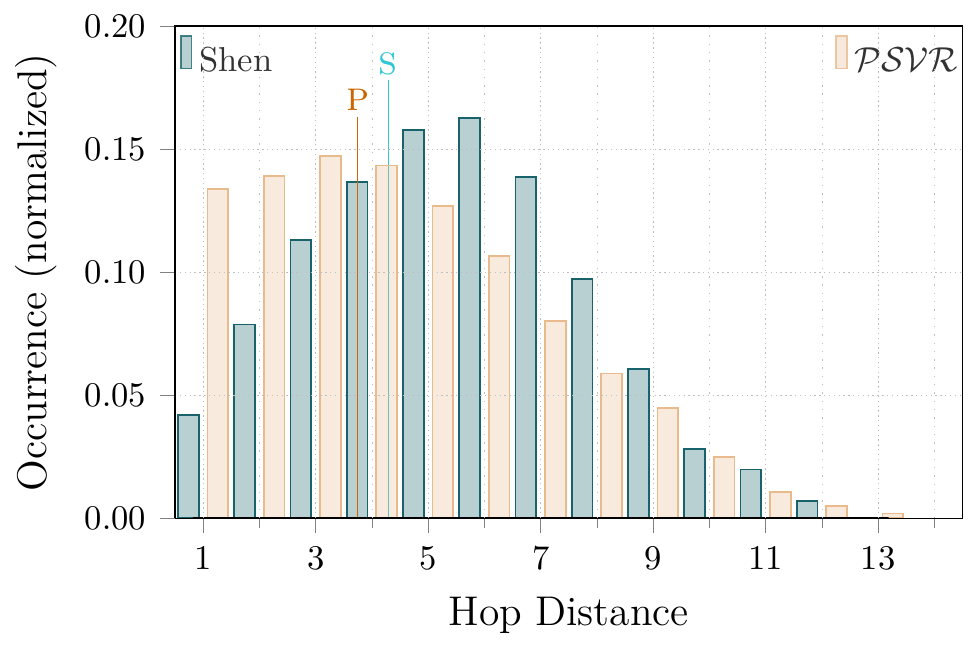}
  \captionof{figure}{Hop distances of delivery paths. Average distance depicted by \textcolor{cyan}{\textit{S}} for Shen's approach and \textcolor{orange}{\textit{P}} for $\mathcal{PSVR}$.}
  \label{fig:appendix:histogram}
\end{minipage}\hfill
\begin{minipage}{.46\textwidth}
  \centering
\includegraphics[width=1.07\textwidth]{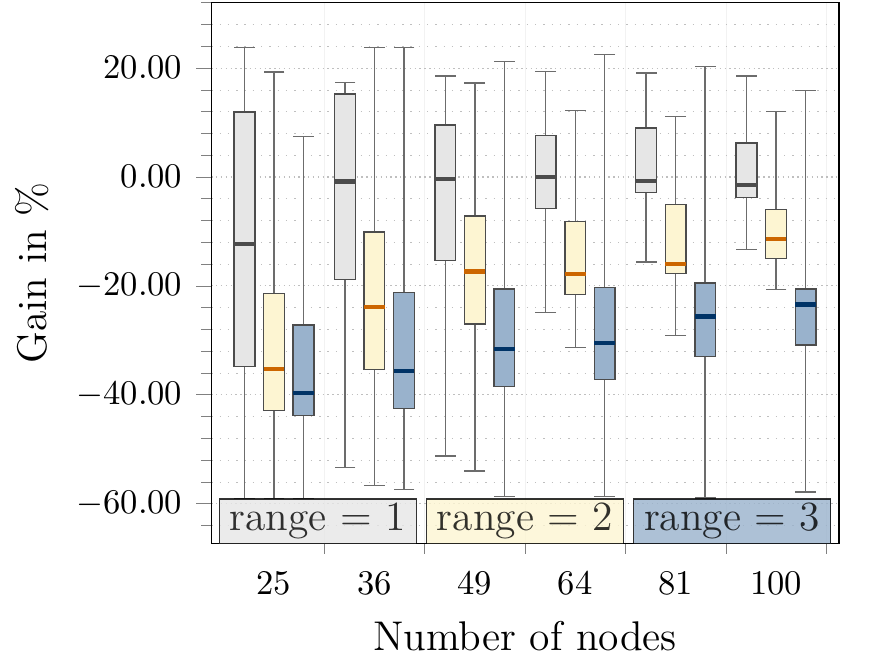}
  \captionof{figure}{Average gain per subscription.}
  \label{fig:appendix:msg-avg-gain}
\end{minipage}
\end{figure}

The gain varies substantially depending on the density (density grows with communication \textit{range}) as can be seen in the boxplots in Fig.\ref{fig:appendix:msg-avg-gain}. The denser the network the more potential shortcuts, hence, the gain in saved
messages is increased. This also holds for the simulations with \OMNET
and the applied radio model (includes path loss and slow fading).

\subsection{Real World Deployment: Throughput and Robustness}

\PSVR delivers all publications while no error in the underlying routing structure occurs. Figure~\ref{fig:RW} shows the delivery ratio in percent for multiple tests on a real sensor network deployment at the Fit-IoT Lab in France~\cite{fambon2014fit}. For each number of nodes 20 tests are conducted each lasting two hours. An initial setup phase of 10 minutes is granted until publication delivery starts. Publications were dispatched every 20s. In Fig.~\ref{fig:RW} (right) the same experiment is run for ten hours. Whenever an error occurs in the network the publication delivery ratio decreases, in error free phases the value can recover. The figure shows a single representative example for 10, 20, and 50 nodes. As the Fit-IoT Lab can be used at the same time by other people, possibly executing bandwidth demanding experiments, a long term test shows the recovery strength of our approach.

\vspace{-3mm}
\begin{figure}[h]
\centering
  \includegraphics[width=1\textwidth]{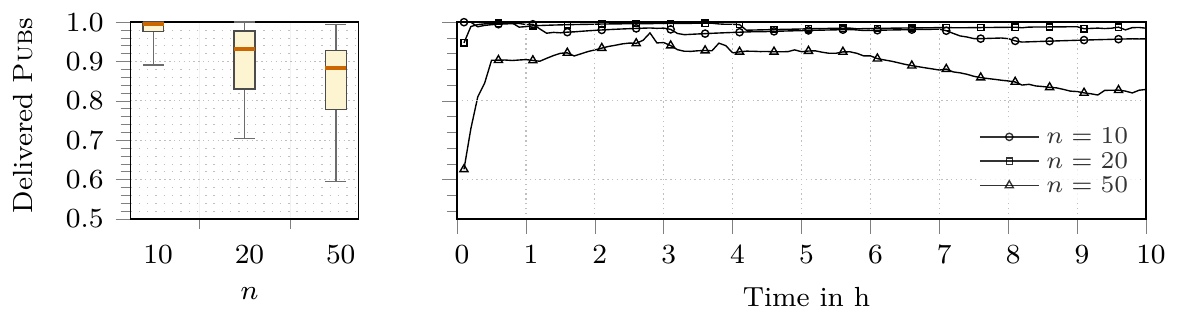}
\caption{Delivered publications average and long term test snap shot.}
\label{fig:RW}
\end{figure}

As can be seen by Fig.~\ref{fig:RW} (left) unsurprisingly an increasing number of nodes is more demanding on the \ps system. In short periods of time and mostly when the wireless channel is in use by other experiments the delivery ratio decreases. As can be seen in Fig.~\ref{fig:RW} (right), the 10 minute setup period was occasionally to short for the 50 nodes experiments also causing a drop in the delivery turnout. On average it stayed in the 80\% to 90\% margin, which we find tolerable considering the benefits of inherent fault tolerance and dynamic adaptability. 


\section*{Conclusion}
\vspace{-2mm}
The presented pub/sub system $\mathcal{PSVR}$ significantly enhances
the algorithm of \cite{Siegemund_VR:2015}. $\mathcal{PSVR}$ is
optimized for scenarios where communications links are unstable and
nodes frequently change subscriptions. It is a compromise of size and
maintenance effort for routing tables due to sub- and unsubscriptions
and the length of routing paths. Simulations and verification against
theoretical, closer to optimal solutions revealed that our approach
gives a fair trade-off between the scalability of the support
structure and the message forwarding overhead. Real world tests
confirmed its usability. The approach scales with the number of nodes
and is suitable for wireless ad-hoc networks.

\vspace{-3mm}
\bibliographystyle{plain}
\bibliography{document}

\end{document}